\documentclass[sigconf]{acmart}
\setcopyright{none}
\renewcommand\footnotetextcopyrightpermission[1]{}
\usepackage{subcaption}
\usepackage[linesnumbered, ruled, vlined]{algorithm2e}
\usepackage{setspace}
\usepackage{xcolor}
\usepackage{xspace}
\usepackage{listings}

\usepackage{subfloat}
\usepackage{url}
\usepackage{soul}
\usepackage{times}
\sethlcolor{lightgray}
\usepackage[export]{adjustbox}
\usepackage{enumitem}

\usepackage[disabled]{util/aplcomments}

\newcommenter{yxy}{1.0,0.0,0.0}
\newcommenter{zhihan}{1.0,0.5,0.0}
\newcommenter{kan}{0.0, 1.0, 0.0}
\newcommenter{cong}{0.0,0.8,0.8}

\definecolor{revised-color}{rgb}{0, 0, 0}
\definecolor{topic-color}{rgb}{0.72, 0.22, 0.22}
\definecolor{edits-color}{rgb}{0.03, 0.60, 0.57}
\newcommand{\revised}[1]{{\color{revised-color} #1}}
\newcommand{\edits}[1]{#1}

\newcommand{\name}{Bamboo\xspace}
\newcommand{\system}{Bamboo\xspace}
\newcommand{\txn}[1]{\texttt{T#1}}
\newcommand{\code}[1]{\tt{#1}}
\newcommand{\ww}{Wound-Wait\xspace}
\newcommand{\waitd}{Wait-Die\xspace}
\newcommand{\owners}{\texttt{owners}\xspace}
\newcommand{\waiters}{\texttt{waiters}\xspace}
\newcommand{\retired}{\texttt{retired}\xspace}
\newcommand{\secref}[1]{Section~\ref{#1}}
\newcommand{\figref}[1]{Figure~\ref{#1}}

\newcommand{\lstref}[1]{Listing~\ref{#1}}
\newcommand{\algoref}[1]{Algorithm~\ref{#1}}

\newtheorem{property}{Property}
\newtheorem{theorem}{Theorem}
\newtheorem{definition}{Definition}
\newtheorem{lemma}{Lemma}
\definecolor{comment-color}{rgb}{0.5,0.1,0.1}
\newcommand{\codeComment}[1]{\textnormal{\color{comment-color}{\textit{\#
#1}}}\unskip}
\let\oldnl\nl
\newcommand{\nonl}{\renewcommand{\nl}{\let\nl\oldnl}}

\AtBeginDocument{%
  \providecommand\BibTeX{{%
    \normalfont B\kern-0.5em{\scshape i\kern-0.25em b}\kern-0.8em\TeX}}}

\begin{document}
\fancyhead{} 
\title{Releasing Locks As Early As You Can: Reducing Contention of Hotspots by Violating Two-Phase Locking (Extended Version)}


\author{Zhihan Guo}
\affiliation{
\institution{University of Wisconsin-Madison}
\city{Madison}
\state{WI}
\country{USA}
}
\email{zhihan@cs.wisc.edu}

\author{Kan Wu}
\affiliation{
\institution{University of Wisconsin-Madison}
\city{Madison}
\state{WI}
\country{USA}
}
\email{kanwu@cs.wisc.edu}

\author{Cong Yan}
\affiliation{
\institution{Microsoft}
\city{Redmond}
\state{WA}
\country{USA}
}
\email{coyan@microsoft.com }

\author{Xiangyao Yu}
\affiliation{
\institution{University of Wisconsin-Madison}
\city{Madison}
\state{WI}
\country{USA}
}
\email{yxy@cs.wisc.edu}


\begin{abstract}
  Hotspots, a small set of tuples frequently read/written by a large number of transactions, cause contention in a concurrency control protocol. 
While a hotspot may comprise only a small fraction of a transaction's execution time, conventional strict two-phase locking allows a transaction to release lock only after the transaction completes, which leaves significant parallelism unexploited. Ideally, a concurrency control protocol serializes transactions only for the duration of the hotspots, rather than the duration of transactions.  

We observe that exploiting such parallelism requires violating two-phase locking. In this paper, we propose \name, a new concurrency control protocol that can enable such parallelism by modifying the conventional two-phase locking, while maintaining the same guarantees in correctness. 
\revised{We thoroughly analyzed the effect of cascading aborts involved in reading uncommitted data and discussed optimizations that can be applied to further improve the performance.} Our evaluation on TPC-C shows a performance improvement up to 4$\times$ compared to the best of pessimistic and optimistic baseline protocols. On synthetic workloads that contain a single hotspot, \name achieves a speedup up to \revised{19$\times$} over baselines.
\end{abstract}





\maketitle

\section{Introduction}
Modern highly contentious transactional workloads suffer from hotspots. A hotspot is one or a small number of database records that are frequently accessed by a large number of concurrent transactions. Conventional concurrency control protocols need to serialize such transactions in order to support strong isolation like serializability, even though the hotspot may comprise only a small fraction of a transaction’s execution time.
\figref{fig:hotspot} illustrates the effect using a single hotspot of tuple $A$. 
For both pessimistic (\figref{fig:hotspot-2pl}) and optimistic (\figref{fig:hotspot-occ}) concurrency control, transactions wait or abort/restart at the granularity of entire transactions. 

Ideally, we want a concurrency control protocol to \textit{serialize transactions only for the duration of the hotspots} (e.g., in \figref{fig:hotspot-ideal}, transaction \txn{2} can access the hotspot immediately after \txn{1} finishes writing it) but execute the rest of the transactions in parallel. If the hotspot comprises only a small fraction of the transaction's runtime, such an ideal protocol can improve performance substantially.

\begin{figure}[t]%
    \centering
    \begin{subfigure}[t]{0.32\columnwidth}
        \includegraphics[width=.75\linewidth,valign=t]{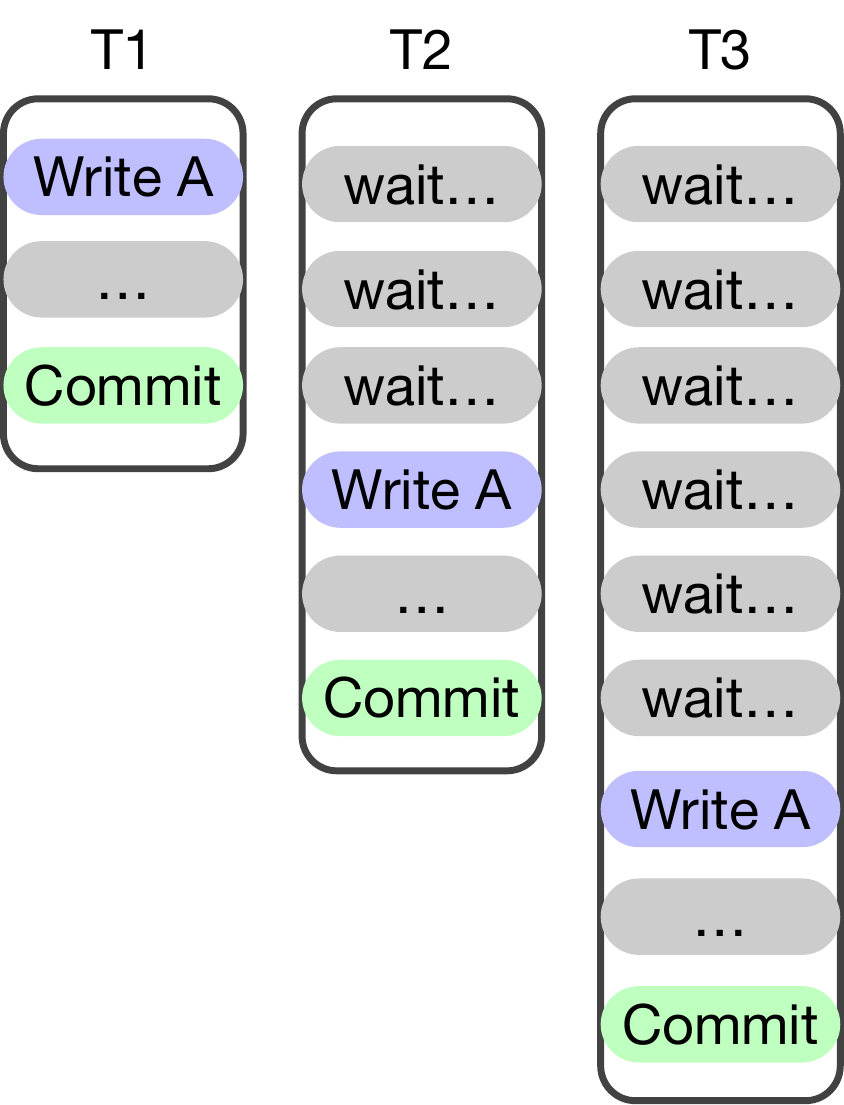} 
        \caption{2PL}
        \label{fig:hotspot-2pl} 
    \end{subfigure}
    \begin{subfigure}[t]{0.32\columnwidth}
        \includegraphics[width=.75\linewidth,valign=t]{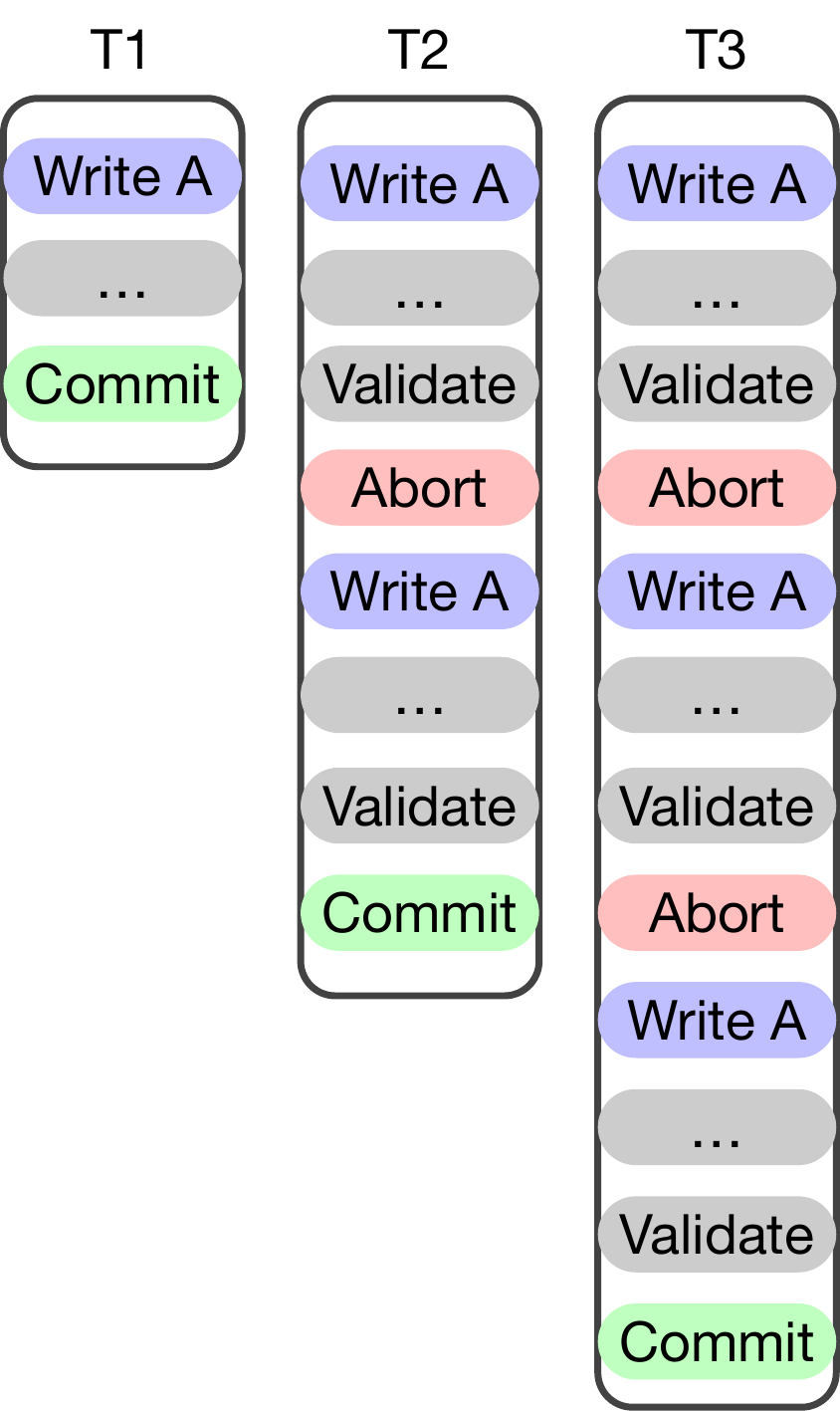} 
        \caption{OCC}
        \label{fig:hotspot-occ} 
    \end{subfigure} 
    \begin{subfigure}[t]{0.32\columnwidth}
        \includegraphics[width=.75\linewidth,valign=t]{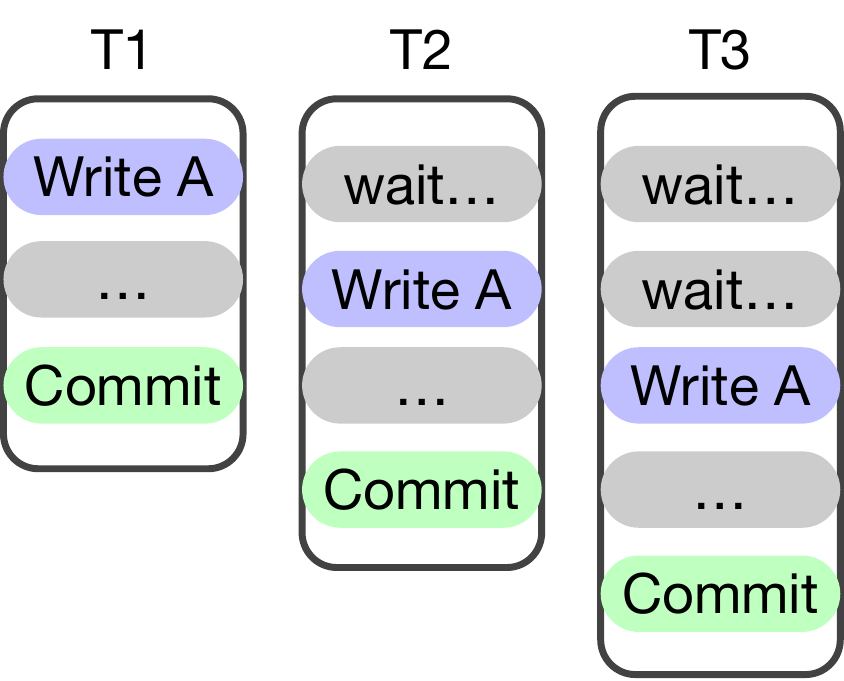} 
        \caption{ideal}
        \label{fig:hotspot-ideal} 
    \end{subfigure} 
    \vspace{-.1in}
    \caption{Schedules of transactions with a hotspot A under 2PL, OCC, and an ideal case. \revised{("Write" means read-modify-write)}} 
    \label{fig:hotspot}%
\end{figure}

\revised{Many production systems (MS Orleans \cite{orleans-txn}, IBM IMS/VS~\cite{gawlick1985varieties}, Hekaton \cite{diaconu13, larson11}, etc.) and research work \cite{agrawal1995ordered, ding2018improving, faleiro15, graefe2013controlled, opt-distributed, johnson2010aether, kimura2012efficient, mu2019deferred, shasha1995transaction, xie2015high, yan2016leveraging} mitigate hotspots by adding extra complication, but cannot achieve the ideal protocol mentioned above.}
In particular, the ideal protocol needs to read dirty data written by another transaction that has not committed yet. For pessimistic concurrency control, this violates the conventional definition of 2PL --- \txn{1} can acquire new locks after releasing locks to other transactions. For OCC and \edits{hybrid concurrency control protocols such as MOCC~\cite{wang2016mostly} and CormCC~\cite{tang2018toward}}, a transaction makes its writes visible only after the execution finishes, which is inherently incompatible with the notion of accessing dirty writes early.

Transaction chopping~\cite{shasha1995transaction,weikum2001transactional,zhang2013transaction} and its recent improvements (e.g., IC3~\cite{wang2016scaling} and Runtime Pipelining~\cite{xie2015high}) are a line of research that tried to enable early reads of dirty data. 
Transaction chopping performs program analysis to decompose a transaction into sub-transactions and allows a sub-transaction to make local updates visible immediately after it finishes. 
While these techniques can substantially improve performance, they have several severe limitations. \textbf{First}, these methods require the \textit{full knowledge of the workload} including the number of transaction templates and the columns/tables each transaction will access. Any new ad-hoc transaction incurs an expensive re-analysis of the entire workload. \textbf{Second}, chopping must follow specific criteria
to avoid deadlocks and ensure serializability, which limits the level of parallelism can be potentially exploited (see \secref{sec:txn-chopping} for details). \textbf{Third}, conservative conflict detections based on the limited information before execution can enforce unnecessary waiting. For example, in IC3, two transactions accessing the same column of different tuples may end up causing contention. 



\edits{In this paper, we aim to explore the design space of allowing dirty reads for general database transactions without the extra assumptions made in transaction chopping.} To this end, we propose \textbf{\name}, a pessimistic concurrency control protocol allowing transactions to read dirty data during the execution phase (thus violating 2PL), while still providing serializability. \name is based on the \ww variant of 2PL and can be easily integrated into existing locking schemes. It allows a transaction to \textit{retire} its lock on a tuple after its last time updating the tuple so that other transactions can access the data. 
Annotations of the last write can be provided by programmers or programming analysis. 
To enforce serializability, \name tracks dependency of dirty reads through the lock table and aborts transactions when the dependency is violated. 

One well-known problem of violating 2PL is the introduction of \textit{cascading aborts}~\cite{agrawal1995ordered} --- an aborted transaction causes all transactions that have read its dirty data to also abort. If not properly controlled, cascading aborts lead to a significant waste of resources and performance degradation. 
Through \name, this paper explores the design space and trade-off of cascading aborts, evaluates its overhead, and proposes optimizations to mitigate these aborts. 
In summary, this paper makes the following contributions.
\begin{itemize}[leftmargin=.2in]
	\item We developed \name, a new concurrency control protocol that violates 2PL to improve parallelism for transactional workloads without requiring the knowledge of the workload ahead of time. \name is provably correct. 
	\item We conducted a thorough analysis (both qualitatively and quantitatively) of the cascading abort effect, and proposed optimizations to mitigate such aborts.
	\item We evaluated \name in the context of both \textit{interactive transactions} and \textit{stored procedures}. \edits{In TPC-C, \name demonstrated a performance improvement up to 2$\times$ for stored procedures and 4$\times$ for interactive transactions compared to the best baseline (i.e., Wait-Die and Silo respectively). \name also outperforms IC3 by 2$\times$ when the attributes of hotspot tuples in TPC-C are truly shared by transactions.}
\end{itemize}

\vspace{-.15in}
\section{Background and Motivation} \label{sec:related}

\secref{sec:2pl} describes the background of two-phase locking, with a special focus on the \ww variant which \name is based on. Then, \edits{\secref{sec:txn-chopping} discusses how transaction chopping mitigates the hotspot issue.}

\subsection{Two-Phase Locking (2PL)} \label{sec:2pl}


Two-phase locking (2PL) is the most widely used class of concurrency control in database systems. In 2PL, reads and writes are synchronized through explicit locks in shared (\textit{SH}) or exclusive (\textit{EX}) mode. A transaction operates on a tuple only if it has become an ``owner'' of the corresponding lock. 

According to Bernstein and Goodman \cite{bernstein1981concurrency}, 2PL forces two rules in acquiring locks: 1) conflicting locks are not allowed at the same time for the same data; 2) a transaction cannot acquire more locks once it releases any. 

The second rule requires every transaction obtaining locks to follow two phases: \textit{growing phase} and \textit{shrinking phase}. A transaction can acquire locks in the growing phase but will enter the shrinking phase if they ever releases a lock. In the shrinking phase, no more locks should be acquired. This rule guarantees serializability of executions by ensuring no cycles of dependency among transactions. 

For a lock request that violates the first rule, 2PL may put the requesting transaction on the waiting queue until the lock is available. 
Two major approaches exist to avoid deadlocks due to cycles of waiting: \textit{deadlock detection} and \textit{deadlock prevention}. The former explicitly maintains a central \textit{wait-for} graph and checks for cycles periodically. The graph becomes a scalability bottleneck with highly parallel modern hardware~\cite{yu2014}. Deadlock prevention technique instead allows waiting only when certain criteria are met. The two mostly popular protocols under this category are \ww and \waitd~\cite{bernstein81, rosenkrantz1978system}. 

\vspace{.05in}
\noindent\textbf{Wound-Wait Variant of 2PL}
\vspace{.05in}

In \ww, each transaction is assigned a timestamp when it starts execution; transactions with smaller timestamps have higher priority. When a conflict occurs, the requesting transaction $T$ compares its own timestamp with the timestamps of the current lock owners --- owners whose timestamps are bigger than $T$ are aborted, namely \texttt{Wound}. Then $T$ either becomes the new owner (i.e., all current owners are aborted) or waits for the lock (i.e., some owners remain), namely \texttt{Wait}. 
The lock entry for each tuple maintains \owners and \waiters as two lists of transactions that are owning or waiting for the lock on the tuple (\figref{fig:data-structure}). \waiters can be sorted based on the transactions' timestamps to simplify the process of moving transactions from \waiters to \owners.

\ww is deadlock-free because a transaction can only wait for other transactions that have smaller timestamps. In the wait-for graph, this means all the edges are from transactions with larger timestamps to transactions with smaller timestamps, which inherently prevents cycles. 
Besides deadlock-freedom, \ww is also starvation-free as the oldest transaction has the highest priority and will never abort due to a conflict. \ww is the concurrency control protocol used in Google Spanner~\cite{corbett12, malkhi2013spanner}.

\vspace{.05in}
\noindent\textbf{Wait-Die Variant of 2PL}
\vspace{.05in}

Different from \ww, when a conflict occurs, \waitd allows transactions with smaller timestamps to wait (i.e., \texttt{Wait}) and transactions with larger timestamps to self-abort (i.e., \texttt{Die}). \waitd is also deadlock- and starvation-free; it is the concurrency control protocol used in Microsoft Orleans~\cite{orleans-txn}.
\vspace{-.15in}

\subsection{Transaction Chopping}
\label{sec:txn-chopping}

\edits{Similar to \name, transaction chopping~\cite{shasha1995transaction} aims to increase concurrency when hotspots are present. In particular, it chops a transaction into smaller sub-transactions and allow an update to be visible after the sub-transaction finishes but before the entire transaction commits. 
In particular, an SC-graph is created based on static analysis of the workload, where each sub-transaction represents a node in the graph. Sub-transactions of the same transaction are connected by sibling (S) edges. Sub-transactions of different transactions are connected by conflict (C) edges if they have potential conflicts. Chopping requires no cycle in the graph and only the first piece can roll-back or abort. It obtains the finest chopping that can guarantee safeness based on static information. 

IC3~\cite{wang2016scaling} is the state-of-the-art concurrency control protocol in this line of research. IC3 achieves fine-grained chopping through column-level static analysis. Sub-transactions accessing different columns of the same table will no longer introduce C-edges. As IC3 allows cycles in the SC-graph, during runtime, it tracks dependencies of transactions and enforces pieces involving C-edges to execute in order to maintain serializability. Moreover, it proposes optimistic execution to enforce waiting just on validation and commit phases for non-conflicting transactions accessing same columns of different tuples. Although inducing more aborts, the optimistic approach still show advantages under high contention. 

However, IC3 still has several limitations. First, it assumes column accesses of all transactions to be known before execution and does not support ad-hoc transactions. Second, chopping must guarantee no crosses of C-edges to avoid potential deadlocks. For example, if one transaction accesses table A before B while the other accesses table B before A. The accesses of table A and B must be merged into one piece, limiting concurrency. Third, column-level static analysis does not exploit the concurrency when transactions accessing same columns of different tuples; it helps reduce more contention only when transactions access different columns of the same tuple. 
We will show quantitative evaluations of IC3 compared with \name in \secref{sec:exp-ic3}. 


}
\section{\name}
\label{sec:protocol}

The basic idea of \name is simple --- in certain controlled circumstances, we allow other transactions to violate an exclusive lock held by a particular transaction. A transaction's dirty updates can be accessed after it has finished its updates on the tuple, following the idea shown in \figref{fig:hotspot-ideal}. 


\subsection{Challenges of Violating 2PL}
\label{sec:challenges}

Although violating 2PL offers great performance potential, it also brings two key challenges that we highlight below. 



\vspace{.05in}
\noindent\textbf{Challenge 1: Dependency Tracking}
\vspace{.05in}

A conventional 2PL protocol uses locks to track dependencies among transactions. 2PL protocols use various techniques (cf. Section~\ref{sec:2pl}) to prevent/break a cycle in the dependency graph. 
We call an edge in a conventional dependency graph a \textit{lock-induced} edge. 

In contrast, \name allows a transaction to read dirty value 
without waiting for locks. Such a \textit{read-after-write} dependency can be part of a cycle and yet is not captured by a conventional lock. For example, \txn{1} may read \txn{2}'s dirty write on record A and \txn{2} may read \txn{1}'s dirty write on record B. Such a cycle is not captured by the ``wait-for'' relationship. We call such a dependency edge a \textit{dirty-read-induced} edge.

For \name to work efficiently, we need a new deadlock avoidance mechanism that can avoid cycles caused by both lock-induced and dirty-read-induced edges uniformly. Section~\ref{sec:protocol-desc} will describe the detailed protocol we build that achieves this goal. 

\vspace{.05in}
\noindent\textbf{Challenge 2: Cascading Aborts}
\vspace{.05in}

Allowing a transaction to read dirty data may lead to cascading aborts, as pointed out in multiple previous protocols~\cite{orleans-txn, larson11}. Specifically, if \txn{2} reads \txn{1}'s update before \txn{1} commits, a commit dependency between the two transactions is established --- \txn{2} is able to commit only if \txn{1} has successfully committed. If \txn{1} decides to abort (e.g., due to conflicts or integrity violation) then all the transactions that have commit dependencies on \txn{1} must also abort; this includes \txn{2} and all the transactions that have read \txn{2}'s dirty writes and so forth. This means potentially a long chain of transactions with commit dependencies need to cascadingly abort, causing waste of work and performance degradation. In Section~\ref{sec:cascade}, we present a deeper analysis of cascading aborts.

\subsection{Protocol Description}
\label{sec:protocol-desc}


This section describes the basic \name protocol in detail. 
In particular, we focus on addressing the first challenge in Section~\ref{sec:challenges} (i.e., dependency tracking). 
\name is developed based on \ww (cf. \secref{sec:2pl}); our description mainly focuses on the differences between the two. 
We firstly describe the new data structures \name requires to track dependencies of dirty reads,
followed by a detailed description of the pseudocode of the protocol.


\subsubsection{Data Structures}

All edges in the dependency graph of a conventional 2PL protocol are lock-induced edges. They are captured by locks and maintained in lock entries of individual tuples. 
For \name, we try to uniformly handle both lock-induced and dirty-read-induced edges by adding extra metadata into each lock entry and transaction.


\begin{figure}[t]%
    \centering
    \includegraphics[width=\linewidth]{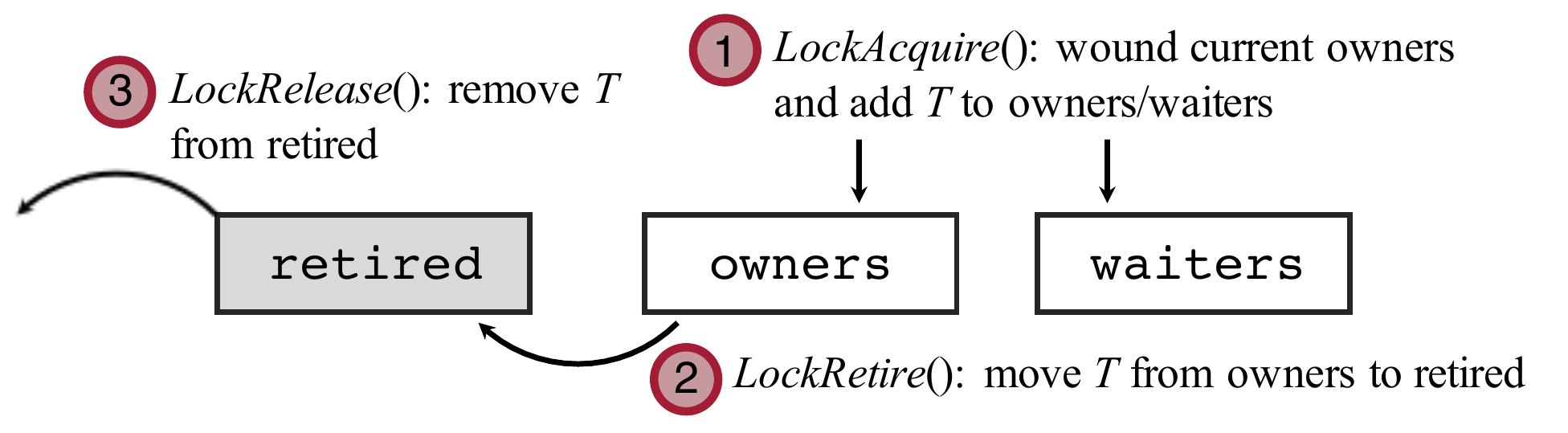} 
    \caption{\textbf{A lock entry in \name} --- \normalfont{Transaction \txn{} moves between lists (i.e., \owners, \waiters, and \retired) through function calls. The \retired list does not exist in baseline \ww.}} 
    \label{fig:data-structure}%
\end{figure}


\textbf{\texttt{tuple.\retired:}}
\name adds a new list called \retired in each lock entry next to the existing lists of \owners and \waiters, as shown in \figref{fig:data-structure}. \retired is sorted based on the timestamps of transactions in it. 
After a transaction has finished updating a tuple, the transaction can be moved from \owners to \retired. This allows other transactions to join \owners and read the dirty updates of the retired transactions. By maintaining the \retired list,
a dirty-read dependency can be captured in the lock entry --- if a retired transaction \txn{} has an exclusive lock, then all transactions in \retired after \txn{} and all transactions in \owners depend on \txn{}. 
Adding \retired allows both lock-induced and dirty-read-induced dependencies to be maintained in the lock entry.

\textbf{\texttt{Transaction.commit\_semaphore:}}
\name uses a new variable \texttt{commit\_semaphore} to ensure that transactions with dirty-read dependencies commit in order. A transaction \txn{} increments its own semaphore when it conflicts with any transaction in \retired of any tuple. The semaphore is decremented only when the dependent transaction leaves \retired so that \txn{} becomes one of the leading non-conflicting transactions in \retired. \revised{The semaphore is implemented using a 
64-bit integer for each transaction and incremented/decremented through atomic operations. The number of accesses to the semaphore is bounded by the number of tuple accesses of 
a transaction. The overhead of the semaphore is within 0.2\% of the total execution time with 120 threads under a high-contention workload. Details of how this variable is operated will be discussed in the concrete protocol. }

\subsubsection{Locking Logic}

\begin{algorithm}[t!]
\setstretch{0.9}
\small
\nonl\codeComment{req\_type is SH or EX} \\
\textit{LockAcquire(txn, req\_type1, tuple1)} \\
\nonl... \\
\nonl\codeComment{Lock can retire after the last write to the tuple} \\
\hl{\textit{LockRetire(txn, tuple1)}} \\
\textit{LockAcquire(txn, req\_type2, tuple2)} \\

\nonl... \\

\nonl\codeComment{Wait for transactions that txn depends on} \\
\While{\hl{txn.commit\_semaphore $\not=$ 0}}{\hl{pause}}
\textit{\revised{if(!abort)} writeLog()} \codeComment{Log to persistent storage device}\\

\textit{LockRelease(txn, tuple1, is\_abort)} \\
\textit{LockRelease(txn, tuple2, is\_abort)} \\
\textit{txn.terminate(is\_abort)}

\caption{\textbf{A transaction's lifecycle in \name} --- Differences between \name and \ww are highlighted in \hl{gray}.}
\label{alg:lifecycle}

\end{algorithm}


\algoref{alg:lifecycle} shows the lifecycle of how the database executes a transaction with \name. 
It largely remains the same as a conventional \ww 2PL protocol, with the differences highlighted using gray background color. For the basic protocol, we assume each transaction is assigned a timestamp when it first started similar to \ww; we will later optimize the timestamp assignment process in \secref{ssec:optimization}.

Different from conventional 2PL, \name allows a transaction to immediately retire a tuple that will not be written again by the transaction (line 2); the transaction can still read the tuple since \name keeps a local copy of the tuple for each read request. The transaction can acquire more locks on other tuples after retiring some lock (line 3). After the execution finishes, a transaction must wait for its \texttt{commit\_semaphore} to become 0 
before it can start committing. 
The transaction then moves forward to perform logging (line 6) and release locks (lines 7--8). Finally, the transaction either commits or aborts depending on the execution result. If an abort occurs during the transaction execution, the transaction directly jumps to line 7 to release locks. 

\begin{algorithm}[t!]
\setstretch{0.9}
\small
\nonl\textit{tuple.retired} \codeComment{List of retired transactions \revised{ordered} by ascending timestamp order} \\
\nonl\textit{tuple.owners} \codeComment{List of owners of the lock} \\
\nonl\textit{tuple.waiters} \codeComment{List of transactions waiting for the lock, sorted by ascending timestamp order}\\
\nonl\codeComment{Each list above contains \{txn, type\} where type is either SH or EX}\\

\SetKwProg{myfun}{Function}{}{}
\nonl\texttt{\\}
\myfun{LockAcquire(txn, req\_type, tuple)}{
    \textit{has\_conflicts = false} \\
    \For{(t, type) in \hl{concat(tuple.retired,} tuple.owners\hl{)}}{
        \If {conflict(req\_type, type)}{
            \textit{has\_conflicts = true}
        }
        \If{has\_conflicts \textbf{and} txn.ts $<$ t.ts}{
            \textit{t.set\_abort()}
        }
    }
    \textit{tuple.waiters.add(txn)} \\
    \textit{PromoteWaiters(tuple)}
}

\nonl\texttt{\\}
\nonl \codeComment{move txn from tuple.owners to tuple.retired} \\
\myfun{\hl{LockRetire(txn, tuple)}}{ \label{alg:retire}
    \hl{\textit{tuple.owners.remove(txn)}} \\
    \hl{\textit{tuple.retired.add(txn)}} \\
    \hl{\textit{PromoteWaiters(tuple)}}
}

\nonl\texttt{\\}
\myfun{LockRelease(txn, tuple, is\_abort)}{ \label{alg:release}
    \hl{\textit{all\_owners = tuple.retired $\cup$ tuple.owners}}\\
    \If{\hl{is\_aborted \textbf{and} txn.getType(tuple) == EX}}{ 
        \hl{\textit{abort all transactions in all\_owners after txn}}
    }
    \textit{remove txn from \hl{tuple.retired or} tuple.owners} \\
    \If{\hl{txn was the head of tuple.retired \textbf{and} conflict(txn.getType(tuple), tuple.retired.head)}}{
        \nonl\codeComment{heads: leading non-conflicting transactions}\\
        \nonl\codeComment{Notify transactions whose dependency is clear} \\
        \For{\hl{t in all\_owners.heads}}{
            \hl{\textit{t.commit\_semaphore$--$}}
        }        
    }
    \textit{PromoteWaiters(tuple)}
}

\nonl\texttt{\\}
\myfun{PromoteWaiters(tuple)}{
    \For{t in tuple.waiters}{
        \If{conflict(t.type, tuple.owners.type)}{
            break \\
        }
        \textit{tuple.waiters.remove(t)} \\
        \textit{tuple.owners.add(t)} \\
        \If{\hl{$\exists$(t', type) $\in$ tuple.retired s.t. conflict(type, t.type)}}{
            \hl{\textit{t.commit\_semaphore} ++ }
        }
    }
}

\caption{\textbf{Function calls in \name }
--- Difference between \name and \ww is highlighted in \hl{gray}.}
\label{alg:protocol}

\end{algorithm}

\algoref{alg:protocol} shows the detailed implementation of the functions in \name, i.e., \textit{LockAcquire()}, \textit{LockRetire()}, \textit{LockRelease()}, as well as an auxiliary function \textit{PromoteWaiters()} which is called by the other three functions. \revised{Note that the first three functions are in critical sections protected by latches, same as other 2PL protocols.}

The baseline \ww is the algorithm in \algoref{alg:protocol} ignoring the code in gray. \name adds extra logic to \textit{LockAcquire()} and \textit{LockRelease()}, and adds a new function \textit{LockRetire()}. In the following, we walk through these functions step-by-step.


\vspace{.05in}
\textbf{\textit{LockAcquire()}}
\vspace{.05in}

As discussed in Section~\ref{sec:2pl}, in \ww when a conflict occurs, the requesting transaction would abort (i.e., \texttt{wound}) current owners that have a bigger timestamp than the requesting transaction (lines 2--7). 
In \name, we need only a small change which is to wound transactions in \textit{both} \owners and \retired (line 3). 
After some or all of the current owners are aborted, some waiting transaction(s) or the requesting transaction may become the new owner. 
In the pseudo code, for brevity, we show this logic by always adding the requesting transaction to the waiter list (line 8) and then try to promote transactions with small timestamps from \waiters to \owners (by calling \textit{PromoteWaiters()} in line 9). In our actual implementation, unnecessary movement between \waiters and \owners are avoided. 

Inside \textit{PromoteWaiters()}, the algorithm scans \waiters in the growing timestamp order (line 24). For each transaction that does not conflict with the current owner(s) (line 25), it is moved from \waiters to \owners (lines 27--28). Otherwise the loop breaks when the first conflict is encountered (line 26). 
In \name, we also need to increment the \textit{commit\_semaphore} for each transaction that just became an owner if it conflicts with any transaction in \retired (lines 29--30). This allows the transaction to be notified when transactions that it depends on have committed. 

\vspace{.05in}
\textbf{\textit{LockRetire()}}
\vspace{.05in}

This function simply moves the transaction from \owners to \retired of the tuple (lines 11--12). It then calls \textit{PromoteWaiters()} to potentially add more transactions to \owners (line 13). It is important to note that the \textit{LockRetire()} function call is completely optional. 
If the function is never called for all transactions, then \name degenerates to \ww. \name also allows any particular transaction to choose whether to call \textit{LockRetire()} on any particular tuple. Such compatibility allows the system to choose between \name and \ww at a very fine granularity. 

\vspace{.05in}
\textbf{\textit{LockRelease()}}
\vspace{.05in}

In \ww, releasing a lock simply means removing the transaction from \owners of the tuple (line 18) and promoting current waiters to become new owners (line 22). In \name, the \textit{LockRelease()} function does more work: (1) handling cascading aborts and (2) notifying other transactions when their dependency is clear. 

Specifically, we define a list \texttt{all\_owners} to be the concatenation of \retired and \owners (line 15). If the releasing transaction decides to abort and its lock on the tuple has type \textit{EX}, then all the transactions in \texttt{all\_owners} after the releasing transaction must abort cascadingly. In \name, these transactions will be notified abort; the abort logic will be later performed by the corresponding worker threads. Note that if the aborting transaction locks the tuple with type \textit{SH}, then cascading aborts are not triggered --- an \textit{SH} lock has no effect on the following transactions. 

The algorithm then removes the transaction from \retired or \owners depending on where it resides (line 18). If the removed transaction was the old head of \retired and has a lock type that conflicts with the new head of \retired (line 19), then the algorithm notifies all the current leading non-conflicting transactions in \retired (i.e., \texttt{heads} of \retired) that their dependency on this tuple is clear, by decrementing their corresponding \textit{commit\_semaphore} (lines 29--30).

\subsection{Deciding Timing for Lock Retire} \label{ssec:decide_retire}
\revised{ 
In principle, every write can be immediately followed by $\textit{lock\_retire()}$ without affecting correctness. If a transaction writes a tuple for a second time after retiring the lock, it can still ensure serializability by simply aborting all transactions that have seen its first write. Performance will not be affected if each transaction updates each tuple only once. 

For better performance, a lock can be retired after the transaction's last write to the tuple if the tuple may be updated more than once by the same transaction. To determine where the last write is, \name can rely on \textit{programmer annotation} or \textit{program analysis} to find the last write and insert $\textit{lock\_retire()}$ after it. In this section, we discuss the latter approach. 
}
Determining the last write to a tuple can be challenging since the position may depend on the query parameters or tuples accessed earlier in the transaction. 
We illustrate the challenge using a transaction snippet shown in~\lstref{code:example-before}, where $op_1$ and $op_2$ (line 1 and 5) both work on some tuple from the same table {\code table1}. In the example, ideally we want to add {\code LockRetire()} immediately after $op_1$, yet we cannot be sure whether the later operation $op_2$ will execute and access the same tuple or not at the desired retire point. To solve this challenge, \system synthesizes a condition to add to the transaction program that dynamically decides whether to retire a lock. The process is described below.



{\bf Program analysis.} \system first performs standard control and data flow analysis~\cite{dataflow-book} to obtain all the control and dataflow dependencies in the transaction program. It inlines all functions to perform inter-procedural analysis, and constructs a single dependency graph for each transaction. 

{\bf Identify queries.} \system next identifies every tuple access by recognizing the database query API calls. It analyzes the query (e.g., the SQL query string passed to the API call) as well as parameters to understand the table involved and the variable that stores the key of the tuple being accessed. We assume most queries access a single tuple by primary key and \system can check potential tuple-level re-access using the key. For non-key-access queries we assume it touches all tuples and detect table-level re-access.

{\bf Synthesizing retire condition.} If an operation $op$ works on a table that is no longer accessed after $op$ in the transaction, then the lock in $op$ can safely retire. Otherwise, \system will synthesize a condition to decide whether to retire the lock. This condition checks any later access will be executed and touching the same tuple as $op$, where an example is shown on line 3 in~\lstref{code:example-after}. In the condition, the lock in $op_1$ can safely retire either when {\code cond} evaluates to false, which means the later access $op_2$ will not happen, or when {\code cond} evaluates to true but the key of {\code tup1} and {\code tup2} are not equal, which means $op_2$ will happened but not touch the same tuple.

To generate such condition, the value of {\code cond} and the key of {\code tup2} must be computed before the {\code LockRetire()} call. To do so, \system traces the data source along the data dependency path of {\code cond} and the key, and then moves any computation on the path that happens later than $op_1$ to an early position, without changing the program semantic. Internally, \system tries to move the computation to every position later than $op_1$, and stops when it finds the earliest one where all the data dependencies hold after the movement. Then it adds the synthesized condition as well as the {\code LockRetire()} call to a position after the computation.
For instance, line 3 in~\lstref{code:example-before} originally computes the key of {\code tup2} late in the transaction. \system moves the computation of {\code tup2.key} to line 2 in~\lstref{code:example-after}, the earliest position after $op_1$ where no data dependency is violated.

\begin{lstlisting}[language=c, basicstyle=\scriptsize\ttfamily,
        escapeinside={@}{@}, numbers=left, caption={A transaction program snippet}, label={code:example-before}]
LockAcquire(table1, tup1, EX); // op1
... // other queries and computations
tup2.key = f(input);
if (cond)
	LockAcquire(table1, tup2, EX); // op2
\end{lstlisting}

\begin{lstlisting}[language=c, basicstyle=\scriptsize\ttfamily,
        escapeinside={@}{@}, numbers=left, caption={An example of synthesized condition from~\lstref{code:example-before}}, label={code:example-after}]
LockAcquire(table1, tup1, EX); // op1
tup2.key = f(input);
if (!cond || (cond && tup1.key!=tup2.key)) //synthesized
	LockRetire(table2, tup2)
... // other queries and computations
if (cond)
	LockAcquire(table1, tup2, EX); // op2
\end{lstlisting}

{\bf Handling loops.} \system performs loop fission to allow synthesizing condition for lock retire in loops.~\lstref{code:loop-before} shows an example. Because the keys used in later accesses are computed in later loop iterations, \system breaks the loop into two parts, as shown in~\lstref{code:loop-after}, where the first loop computes all the keys and the second loop executes the tuple accesses. \system then adds a nested loop to produce the retire condition, as shown from lin 6-9. This nested loop checks the keys of the rest of the iterations and sets the variable added by \system, {\code can\_retire}, if any later key is the same as the key in the current iteration.  
\system only handles {\code for} loops where the number of iteration is fixed (i.e., not changed inside the loop). For other types of loop, we do not retire locks inside the loop for now and leave it as future work.

\begin{lstlisting}[language=c, basicstyle=\scriptsize\ttfamily,
        escapeinside={@}{@}, numbers=left, caption={A transaction snippet involving for loop}, label={code:loop-before}]
for(i=0; i<input1; i++) {
  key[i] = f(input2[i]);
	tup.key = key[i];
  LockAquire(table, tup, EX);
}
\end{lstlisting}

\begin{lstlisting}[language=c, basicstyle=\scriptsize\ttfamily,
        escapeinside={@}{@}, numbers=left, caption={An example of synthesized condition from~\lstref{code:loop-before}}, label={code:loop-after}]
for(i=0; i<input1; i++) 
  key[i] = f(input2[i]);
for(i=0; i<input1; i++) {
  tup.key = key[i];
  LockAquire(table, tup, EX);
  bool can_retire = true;
	for(j=i+1; j<input1; j++)
    can_retire &&= (key[j]!=tup.key);
  if (can_retire) LockRetire(table, tup);
}
\end{lstlisting}

\vspace{-.1in}
\subsection{Discussions}

This section discusses a few other important aspects of \name. 
An important feature of \name is the strong compatibility with \ww, meaning an existing database can be extended to use \name without a major rewrite. Many important design aspects can be directly inherited from 2PL. 

\textbf{Fault Tolerance:} \name does not require special treatment of logging. As shown in \algoref{alg:lifecycle}, a transaction 
\revised{does not log its commit record until it has satisfied the concurrency control protocol.}
This is  similar to conventional 2PL which logs in the same way.

\textbf{Phantom Protection:} Phantom protection~\cite{eswaran76} in \name uses the same mechanism as in other 2PL protocols, namely, \textit{next-key locking}~\cite{mohan1989aries} in indexes; this technique achieves the same effect as \textit{predicate locking} but is more widely used in practice. In this context, lock retiring can also be applied to inserts or deletes to an index in the same way as read/writes to tuples. 

\textbf{Weak Isolation:} \name can support isolation levels weaker than serializability. For example, \textit{repeatable read} is supported by giving up phantom protection; \textit{read committed} (RC) is supported by releasing shared locks early. \revised{For RC, \name needs to retire only writes since read locks are always immediately released}. Finally, \textit{read uncommitted} means each retire becomes a release. 

\textbf{Other Variants of 2PL:} To ensure deadlock freedom, \name permits \txn{2} to read \txn{1}'s dirty write only if such a dependency edge is  permitted in the underlying 2PL protocol. \name can be extended to any variants of 2PL but some variants fit better than others. \waitd, for example, allows only older transactions to wait for younger transactions. When applying retiring and dirty reads to this setting, the older transactions are subject to cascading aborts, meaning an unlucky old transaction may starve and never commit. Such problems do not exist in \ww.

\textbf{Compatibility with Underlying 2PL:}
As we pointed out in \secref{sec:protocol-desc}, it is possible to smoothly transition between \name and the underlying 2PL protocol. Specifically, the \textit{LockRetire()} function call is completely optional for any transaction on any tuple. When it is not called, dirty reads are disabled for the particular lock and the system behavior degenerates to 2PL. This allows a system to dynamically turn on/off the dirty read optimization of \name based on the performance and frequency of cascading aborts. 

\revised{
\textbf{Opacity:} Opacity is a property that reads must be consistent even before commit. By definition, ensuring opacity means a transaction is not allowed to read uncommitted data. 
If opacity is required for a transaction, \name can enforce it by running the transaction in Wound-Wait (i.e., wait on a tuple until the retired and owners lists are empty). Note that some production systems~\cite{orleans-txn, larson11, diaconu13, corbett12, malkhi2013spanner} also do not support further optimizations for transactions with opacity. }
\vspace{-.05in}
\subsection{Optimizations} \label{ssec:optimization}

Here we introduce four optimizations for \name --- 
the first two are to reduce extra overheads and the rest are to reduce aborts. \revised{These ideas are not entirely new but we discuss them here since they can substantially improve the performance of \name. We also apply them to baseline protocols when applicable.} 

\textbf{Optimization 1: No extra latches for read operations.} In \name, read operations retire automatically in \texttt{LockAcquire()}. We keep a local copy for every new read \revised{unless the data is written by the transaction itself}. A read operation can be moved to the retired list directly whenever it can become the owner. This optimization requires no extra latches for retiring and will not cause more aborts since aborting a transaction holding a read lock does not cause cascading aborts. 
\revised{Although copying may incur extra overhead, existing work shows such overhead scales linearly with the core count~\cite{xie2015high} and the overhead is less than 0.1\% of the runtime with 120 cores under high contention. With the optimization, the cost of accessing extra latch is within 0.8\% of the total execution time with the same setting.}

\revised{\textbf{Optimization 2: No retire when there is no benefit.} 
Write operations do not need to retire if they bring little benefit but increase the chances of cascading abort or the overhead of acquiring latches; read operations can always safely retire since they cannot cause cascading effects. Different heuristics can be used to decide which writes may or may not be retired. 
We use a simple heuristic where writes in the last $\delta$ ($0\leq\delta\leq1$) fraction of accesses are not retired. The intuition is that hotspots at the end of a transaction should not cause long blocking and retiring them has no benefit. 
However, if a transaction turns out to spend significant time (i.e., longer than $\delta$ of the total execution time) waiting on the \texttt{commit\_semaphore}, we will retire those write operations at the end of a transaction. 
}


\textbf{Optimization 3: Eliminate aborts due to read-after-write conflicts.} In the basic \name protocol, when a transaction tries to acquire a shared lock, it needs to abort all the write operations of low-priority transactions in both the retired list and the owner list. We observed that such aborts are unnecessary. As \name keeps local copies for reads and writes, such a new read-operation can read the local copy of other operations without triggering aborts. 
\revised{The optimization naturally fits \name as it allows for read-modify-write over dirty data and multiple uncommitted updates can exist on a tuple. However, the idea cannot be easily applied to existing 2PL as reading uncommitted data is not allowed hence there exists only one copy of the data. There are existing works~\cite{larson11, sadoghi2014} sharing similar ideas that also allow a transaction to choose which version to read as they allow reading uncommitted data under certain scenarios, though they typically support up to one uncommitted version. 
} 



\textbf{Optimization 4: Assign timestamps to a transaction on its first conflict.} 
\revised{Here we explain how to extend \name to support an existing idea, dynamic timestamp assignment, to avoid aborting on the first conflict.} The pseudocode is shown in \algoref{alg:dynamic_ts}. Specifically, lines 1--5 in \algoref{alg:dynamic_ts} are inserted to the beginning of function \textit{LockAcquire()} in \algoref{alg:protocol}. If the incoming transaction conflicts with any other transaction in \retired, \owners, or \waiters (line 1--2), we assign timestamps for all transactions in the three lists, in the order specified by the algorithm (lines 3--4), and then assign timestamp for the incoming transaction (line 5). Timestamp assignment is a single compare\_and\_swap() call (lines 6--8). \revised{\name may be further improved through more complex timestamping strategies~\cite{lomet2012multi}.}
\begin{algorithm}[t!]
\small
\textit{all\_txns = concat(tuple.retired, tuple.owners, tuple.watiers)}\\
\If{txn conflicts with any transaction in all\_txns}{
    \For{t in all\_txns}{
        \textit{set\_ts\_if\_unassgined(t)}
    }
    \textit{set\_ts\_if\_unassigned(txn)} 
}

\SetKwProg{myfun}{Function}{}{}
\nonl\texttt{\\}
\myfun{set\_ts\_if\_unassigned(txn)}{
    \If{txn.ts = UNASSIGNED}{
        atomic\_compare\_and\_swap(\&txn.ts, UNASSIGNED, atomic\_add(global\_ts))
    }
}
\caption{\textbf{Support for dynamic timestamp assignment} --- lines 1--5 are added to the beginning of function \textit{LockAcquire()} in \algoref{alg:protocol}.}
\label{alg:dynamic_ts}
\end{algorithm}

\subsection{Proof of Correctness} \label{sec:proof}
According to the serializability theory, a schedule of transactions is serializable if and only if the serialization graph (SG)\footnote{\scriptsize A serialization graph (SG) is a directed graph whose nodes are committed transactions and whose edges represent conflicts among pairs of transactions.} is acyclic~\cite{bernstein1987concurrency, bernstein79}. In the following, we will first revisit how serializability was proved for 2PL and then develop a similar proof for \name. 
\vspace{-.05in}
\begin{property}\label{property:2pl}
\textbf{[Two-Phase Rule]} A transaction does not acquire more locks once it has released a lock. 
\end{property}
\vspace{-.05in}
Property~\ref{property:2pl} describes the basic rule that all 2PL protocols must follow. This property together with the behavior of locking ensures that all 2PL schedules of committed transactions are serializable. 
\vspace{-.05in}
\begin{theorem}\label{thm:2pl-sr}
Every schedule in 2PL is serializable.
\end{theorem}

\begin{proof}
\vspace{-.05in}
If the serialization graph contains an edge $T_i \rightarrow T_j$, then the two transactions must have a conflict on some tuple $x$, on which $T_j$ acquired a lock after $T_i$ released a lock. 

A graph with a cycle must contain a path that starts and ends with the same transaction, e.g., $T_i \rightarrow T_j  \rightarrow ... \rightarrow T_i$. This means $T_i$ has released a lock before $T_j$ acquires a lock which happens before $T_j$ releases a lock (according to two-phase rule). All of these happen before $T_i$ acquires a lock (according to the last edge). However, this means $T_i$ acquires a lock after it has released a lock, violating the two-phase rule. 
\end{proof}

With \name, we cannot prove for serializability following the exact same proof above. Since for \name, an edge $T_i \rightarrow T_j$ does not imply $T_i$ releasing the lock before $T_j$ acquiring the lock. We introduce the concept of \textit{commit point} to conduct the proof.

\begin{definition}\label{def:commit-point}
\textbf{[Commit Point]} A transaction's commit point is a point in time between the transaction completes all of the operations and recorded the operations in log. (In Algorithm~\ref{alg:lifecycle}, the commit point is after finishing line 6 but before starting line 7.) 
\end{definition}

\begin{lemma}\label{lemma:commit-point}
\textbf{[Commit Point Ordering]} In \name, if the serialization graph contains $T_i \rightarrow T_j$, then $T_j$ reaches the commit point after $T_i$. 
\end{lemma}

\begin{proof}
Without loss of generality, we assume $T_i$ and $T_j$ conflict on tuple $x$. 
We know that $T_j$ can reach its commit point only after $T_j.$\textit{commit\_semaphore} becomes 0 (lines 4--5 in Algorithm~\ref{alg:lifecycle}). 
According to Algorithm~\ref{alg:protocol}, $T_j$\textit{.commit\_semaphore} can become $0$ only after $T_i$ has released its lock on $x$, which can happen only after $T_i$ has reached its commit point (according to Algorithm~\ref{alg:lifecycle}). Together, this means $T_j$ reaches its commit point after $T_i$. 
\end{proof}

\begin{theorem}\label{thm:2pl-sr}
Every schedule in \name is serializable.
\end{theorem}
\begin{proof}
According to Lemma~\ref{lemma:commit-point}, every edge $T_i \rightarrow T_j$ in the serialization graph means $T_j$ reaching the commit point after $T_i$. Therefore, no cycle may exist since a transaction cannot reach the commit point after it has already reached the commit point, finishing the proof.
\end{proof}

Note that the proof for \name above can also be used to prove for serializability of 2PL. The key here is that all transactions reach the commit point following the data dependency order. Even if a transaction $T_j$ reads dirty data from transaction $T_i$, $T_j$ can only reach the commit point after $T_i$ does. 

Similar to the original 2PL proof, the proof for \name only shows serializability for committed transactions. The fact that \name is deadlock-free follows the deadlock-freedom of \ww, the proof of which is beyond the scope of this paper. 
\section{Cascading Aborts}
\label{sec:cascade}
This section analyzes the effect of cascading aborts qualitatively and discusses an optimization that we propose to mitigate the effect.
\vspace{-.1in}
\subsection{Cases Inducing Cascading Aborts} \label{sec:abort-cases}

\textit{Cascading aborts} (also called \textit{cascading rollback}) is a situation where the abort of a transaction causes other transactions to abort. 
In \name, transactions can read uncommitted data, if the transaction that wrote the data aborts, all dependent transactions must also abort cascadingly. 
In our algorithm, the procedure of cascading aborts corresponds to the line 17 of \algoref{alg:protocol}. 

In \name, a transaction $T$ may abort in three cases: (1) when $T$ is wounded by another transaction with higher priority \revised{to prevent deadlocks}, (2) when a transaction that $T$ depends on aborts and $T$ aborts cascadingly, or (3) when $T$ self-aborts due to transactional logic (e.g., inventory level below 0) or user intervention. 
Since cases (1) and (3) can occur in the baseline \ww protocol as well, we focus on discussing the difference between case (2) and the other two cases. 

\vspace{-.15in}

\subsection{Effects of Cascading Aborts as a Trade-off of Reducing Blocking} \label{ssec:cascading-effect}

\revised{The effect of cascading aborts can be evaluated through three metrics --- length of abort chain, abort rate, and abort time. The \textit{length of abort chain} depicts the number of transactions that must abort cascadingly due to one transaction's abort; our empirical result shows the number can be as large as the number of concurrent transactions at high contention. The \textit{abort time} describes} the total CPU time wasted on executing transactions that aborted in the end.  The throughput of a transaction processing system is largely determined by the number of CPU cycles performing useful vs. useless work; \emph{abort time} is an example of such useless work. The other example is the time spent on waiting for a lock, which we defined as \emph{Wait Time}. Unlike the first two indicating the magnitude of the effect, the time measurements illustrate the tradeoff between waits and aborts in a more direct way. \revised{We use this metric to show the trade-off in the evaluation section.}

Compared to \ww, \name substantially reduces the \emph{wait time} but increases \emph{abort time}. Although \name may have more aborts than \ww due to cascading aborts, trading waits for aborts may be a good deal in many cases. 

\revised{There are two reasons why aborts may be preferred in certain cases.} First, with hotspots, all the transactions aborted cascadingly are those that have speculatively read dirty data. These transactions would have been waiting in \ww without making forward progress in the first place. A large portion of cycles wasted on cascading aborts in \name would also be wasted on waiting in \ww. Second, even if a transaction aborts, it warms up the CPU cache with accessed tuples, so subsequent executions become faster~\cite{kung1981optimistic, tu13}. We observe this effects on Silo, an OCC-based protocol, as described in \secref{sec:exp} --- Silo has higher abort rate than many 2PL protocols but also higher throughput.

We build a model based on previous theoretical analysis on 2PL~\cite{gray1992book, Gray1981ASM, bernstein2009principles} to illustrate when the benefits of \name outweigh its overhead. We define $K$ as the number of lock requests per transaction, $N$ as the number of transactions running concurrently, $D$ as the number of data items, and $t$ as the average time spent between lock requests. 
The throughput is proportional to $\frac{N}{(K+1)t}\times(1-AP_{\textit{conflict}}-BP_{\textit{abort}})$,
where $P_{\textit{conflict}}$ and $P_{\textit{abort}}$ denote the probability a transaction encounters a conflict and an abort, respectively; $A$ denotes the fraction of execution time that a transaction spends waiting given a conflict; $B$ denotes the fraction of time spent on aborted execution. 
\name (bb) can reduce $AP_{\textit{conflict}}$ (due to early retire) but increase $BP_{\textit{abort}}$ (due to cascading aborts). It has positive gains over Wound-Wait (ww) when the benefits outweigh the overhead. 

The gain in $AP_{\textit{conflict}}$ is $(A_{\textit{bb}}-A_{\textit{ww}}) P_{\textit{conflict}}$. $P_{\textit{conflict}}$ is a property of the workload and is approximately $NK^2/(2D)$ in both protocols. Specifically, as there are $N-1$ transactions running concurrently holding $NK/2$ locks on average, the probability of a single lock request conflicting with others is $NK/(2D)$ given a uniformly random distribution of accesses. Thus, the probability of a transaction encountering a conflict during its lifetime is $1 - (1 - NK/2D)^K \approx NK^2/(2D)$~\cite{gray1992book, Gray1981ASM}; 
$A_{\textit{bb}}$ is approximately $1/(K+1)$ (i.e., wait for only the duration of one access) and $A_{\textit{ww}}$ is on average $1/2$ (i.e., wait for half of the transaction execution time).

To model $BP_{\textit{abort}}$, we observe that \name and \ww share two common sources of aborts, i.e. aborts due to deadlock and user-initiated aborts. \name introduces another source of aborts due to cascading, represented as $B P_{\textit{cas\_abort}}$, where $P_{\textit{cas\_abort}}$ is the probability that a transaction aborts cascadingly. 
We calculate an upper bound of this cost. We can bound $B$ by $1$ and bound $P_{\textit{cas\_abort}}$ by $(1-P_{\textit{deadlock}})\times P_{\textit{conflict}} \times P_{\textit{deadlock}} \times (N-1)$ (i.e., the current transaction experiences a conflict while some other transaction experiences a deadlock), which is bounded by $NP_{\textit{conflict}}P_{\textit{deadlock}}$. The value of $P_{\textit{deadlock}}$ is approximately $NK^4/4D^2$, which is approximated by the probability of a transaction conflicting with a transaction that is already conflicted with it~\cite{gray1992book, Gray1981ASM}. 

Combining the above, \name has performance advantage when $(A_{\textit{ww}}-A_{\textit{bb}})P_{\textit{conflict}}>BP_{\textit{cas\_abort}}$, which is satisfied when $(\frac{1}{2}-\frac{1}{K+1})P_{\textit{conflict}} > NP_{\textit{conflict}}P_{\textit{deadlock}}$, which is satisfied when $\frac{N^2K^4}{2D^2} < \frac{K-1}{K+1}$. 
For most databases, the data size $D$ is orders of magnitude larger than $N$ and $K$; so the equation will hold. The high-level intuition here is that the probability of a deadlock is much lower than the probability of a conflict~\cite{gray1992book, Gray1981ASM}. \name optimizes for the common case by reducing the cost of a conflict and sacrifices performance of the conner case by increasing the cost of aborts during deadlocks. 

\revised{In \secref{sec:exp}, we performed quantitative evaluations on the impact of cascading aborts and the tradeoff between aborts and waits using the metrics described above. We will  show how the evaluation results corroborate the arguments and modeling here.}


\section{Experimental Evaluation}
\label{sec:exp}
This section evaluates the performance of \name. 
We first introduce the experimental setup in Section~\ref{ssec:exp-setup}, followed by demonstrations of the performance of \name without cascading aborts in Section~\ref{ssec:exp-1hs}. In \secref{ssec:exp-cascading}, we evaluate \name under different scenarios to understand the effect of cascading aborts. We then report the performance of \name on YCSB and TPC-C workloads in Sections~\ref{ssec:ycsb} and ~\ref{ssec:tpcc} respectively to evaluate \name on different distributions and workloads with user-initiated aborts. Finally, Section~\ref{sec:exp-ic3} compares \name and IC3 in performance with TPC-C at high contention. 

\subsection{Experimental Setup} \label{ssec:exp-setup}
We implement \name in DBx1000~\cite{dbx1000, yu2014}, a multi-threaded, in-memory DBMS prototype. DBx1000 stores all data in a row-oriented manner with hash table indexes. The code is open sourced~\cite{code}.
In this paper, we extended DBx1000 to run transactions in both stored-procedure and interactive modes. \revised{In the stored-procedure mode, all accesses in a transaction and the execution logic are ready before execution.

The interactive mode involves two types of nodes: (1) the \emph{DB server} processes requests like \texttt{get\_row()}, \texttt{update\_row()}, and \texttt{commit()}, and (2) the \emph{client server} executes transaction logic and sends requests to the DB server through gRPC. 
As \name does not require knowing the position of the last write for correctness (cf. \secref{ssec:decide_retire}). When \texttt{update\_row()} is called, the DB server immediately retires the lock after the write, essentially treating every write as the last write. In the interactive mode, the second optimization of no retiring does not apply. 
}

DBx1000 includes a pluggable lock manager that supports different concurrency control schemes. This allows us to compare Bamboo with various baselines {\bf within the same system}. We implemented 5 approaches described as follow:
\vspace{-.05in}
\begin{itemize}[leftmargin=.2in]
	\item {\bf WOUND\_WAIT}~\cite{bernstein81}: The \ww variant of 2PL (\secref{sec:2pl}). 
	\item {\bf NO\_WAIT}~\cite{bernstein81}: The No-Wait variant of 2PL where any conflict causes the requesting transaction to abort. 
	\item  {\bf WAIT\_DIE}~\cite{bernstein81}: The \waitd variant of 2PL (\secref{sec:2pl}).
	\item {\bf SILO}~\cite{tu13}: An in-memory database for fast and scalable transaction processing. It implements a variant of OCC. 
	\item {\bf IC3}~\cite{wang2016scaling}: State-of-the-art transaction chopping-based concurrency control protocol as described in \secref{sec:txn-chopping}. 
\end{itemize}
\vspace{-.05in}
\revised{Experiments in stored-procedure mode were run on a machine with four Intel Xeon CPU at 2.8GHz (15 cores) with 1056GB of DRAM, running Ubuntu 16.04. Each core supports two hardware threads. For the interactive mode, experiments were run on workstations provided by cloudlab \cite{cloudlab} with each machine containing} two Intel Xeon CPU at 2.6GHz (32 cores) with 376GB of DRAM, running Ubuntu 16.04. Each core supports two hardware threads. We collect transaction statistics, such as throughput, latency, and abort rates by running each workload for at least 30 seconds. 

\revised{In this paper, we assume that each hotspot contains one tuple and treat a set of hot tuples as multiple hotspots. For the experiments, transactions log to main memory --- modern non-volatile memory would offer similar performance. \name applies all the optimizations introduced in~\secref{ssec:optimization}. 
To decide the choice of $\delta$, we ran microbenchmark with a wide range of $\delta$. In general, as $\delta$ increases, the overhead in \name decreases, which improves performance in low-contention cases. However, a larger $\delta$ also increases the time spent on waiting for locks under high contention, which leads to less than 13\% drop in performance in our experiments. 
To balance the contrary effects under different workloads, we chose a $\delta$ of 0.15 across all workloads. 
As the dynamic timestamp assignment can also be applied to other 2PL-based protocols, we turn on the optimizations whenever they gain improvements from it. However, as only \name involves cascading aborts, the other protocols barely benefit from the optimization.}

\subsection{Experimental Analysis on Bamboo without Cascading Aborts (Single Hotspot)} \label{ssec:exp-1hs}

In this section, we evaluate the potential benefits of \name in the ideal cases where only one hotspot is present and thus does not induce cascading aborts. 

\vspace{.05in}
\noindent\textbf{Single Hotspot at Beginning}
\vspace{.05in}

We firstly design a synthetic workload with all random reads but a single \revised{read-modify-write} hotspot at the beginning. 
In stored-procedure mode, \name shows 6$\times$ improvements against the best-performing 2PL-based protocols (Wait-Die) due to savings on waiting. In the interactive mode, \name is up to 7$\times$ better than the best baseline (Wound-Wait). 

  
\vspace{.05in}
\noindent\textbf{Varying Transaction Length}
\vspace{.05in}

\begin{figure}[t]%
    \begin{subfigure}[t]{0.48\linewidth}
    \includegraphics[width=0.85\linewidth]{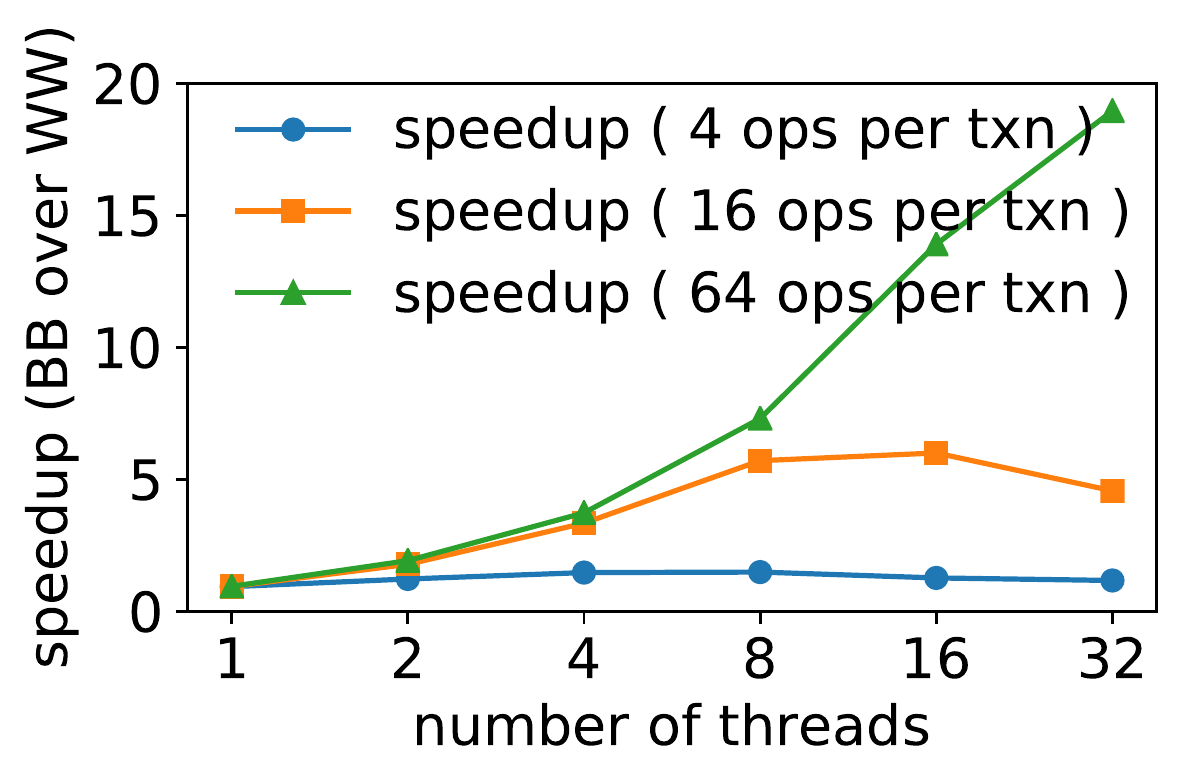} 
    \caption{Varying transaction length}
    \label{fig:one_hs_vs_threads}%
    \end{subfigure}
    \begin{subfigure}[t]{0.03\linewidth}
    \includegraphics[width=\linewidth]{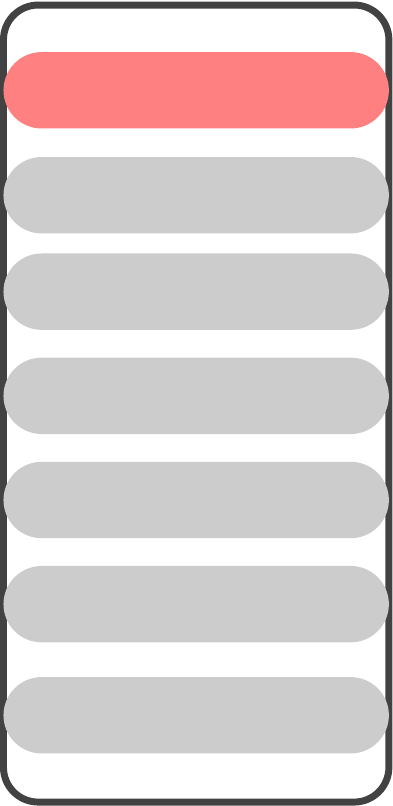} 
	\end{subfigure}
	\hspace{-.05in}
    \begin{subfigure}[t]{0.43\linewidth}
    \centering
    \includegraphics[width=0.9\linewidth]{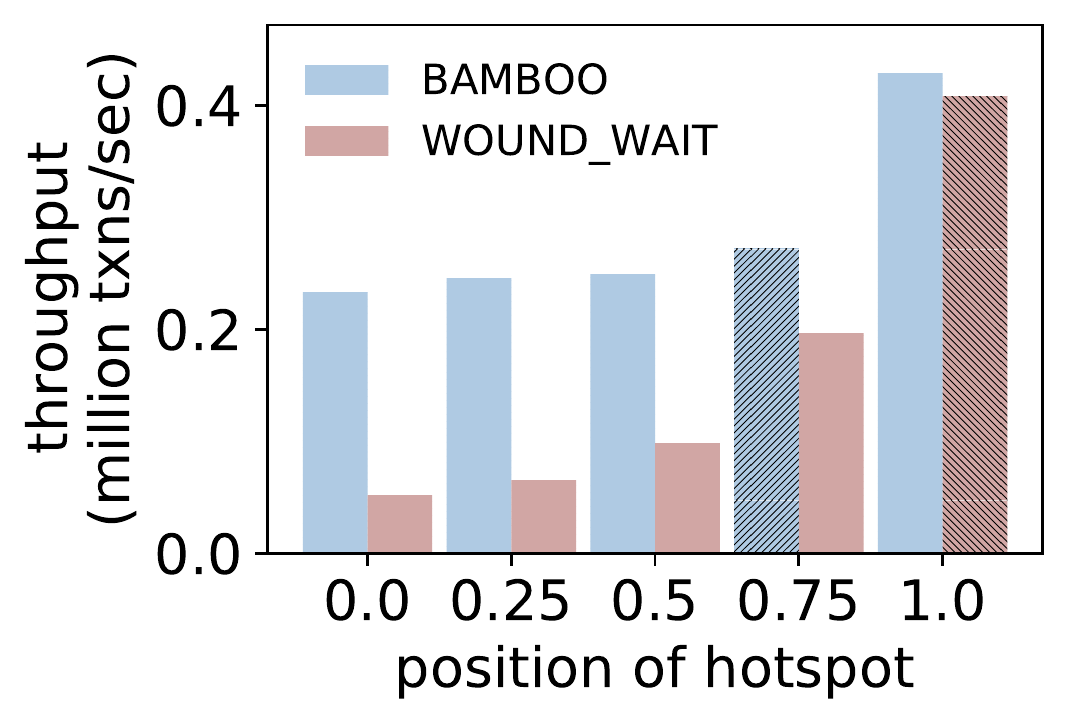} 
    \caption{Varying hotspot position}
    \label{fig:one_hs_vs_hs_pos}%
	\end{subfigure}
	\hspace{-.05in}
	\begin{subfigure}[t]{0.03\linewidth}
    \includegraphics[width=\linewidth]{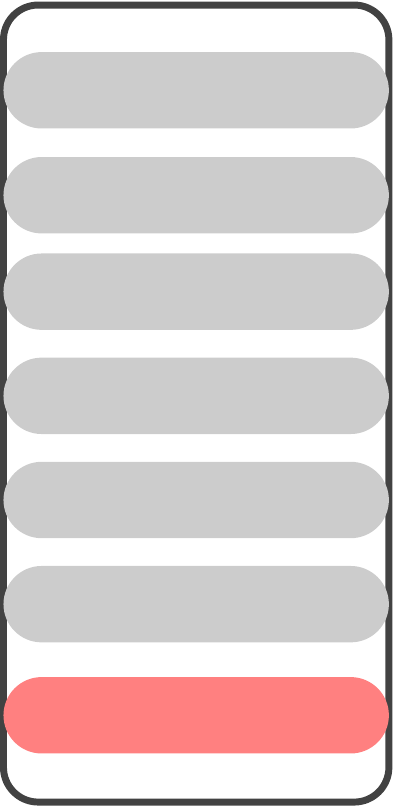} 
	\end{subfigure}
	\vspace{-.1in}
	\caption{Performance on synthetic benchmark with one hotspot at the beginning with various settings, stored-procedure mode}
\end{figure}

\revised{In Figure~\ref{fig:one_hs_vs_threads}, we vary the length of the transactions and report the speedup of \name (BB) over \ww (WW). Firstly, the results shows that \name has greater speedup for longer transactions by up to 19$\times$, which corresponds to a larger $A$ increasing the benefit in the modeling shown in~\secref{ssec:cascading-effect}. }
Secondly, the speedup first increases as the number of threads increases, and then saturates or drops when even more threads come in. This is due to the limitation of the inherent level of parallelism in the workload. 

\vspace{.05in}
\noindent\textbf{Varying Hotspot Position}
\vspace{.05in}

\revised{Instead of fixing the hotspot at the beginning, we vary the position of the hotspots in this experiment. The result is shown in Figure~\ref{fig:one_hs_vs_hs_pos}.}
The two small figures on each side of the x-axis corresponding to the position of the hotspot when $x=0$ (beginning of transaction) and $x=1$ (end of transaction), respectively. We show the results for workloads with transactions of 16 operations, while the results for other transaction lengths have similar observations. \name provides a higher speedup against \ww when the access of hotspot is earlier in a transaction. \revised{The result also aligns with the modeling as an early access gives larger $A_{ww}$ and thus greater benefit from $A_{ww} - A_{bb}$. }

{\bf Summary:} Through reducing lock waiting time, \name can improve performance significantly (up to 19$\times$ over \ww). Some factors have an impact on the performance gain --- \name shows larger speedup under a higher level of parallelism (i.e., more threads), longer transactions, and ``earlier'' hotspot accesses.


\subsection{Experimental Analysis on Bamboo with Cascading Aborts (Multiple Hotspots)}\label{ssec:exp-cascading}

Next, we present empirical evaluation of \name serving workloads that can induce cascading aborts to understand their effects. We start with synthetic workloads with two \revised{read-modify-write} hotspots and fourteen random reads. \revised{Here we use a dataset of more than 100 GB.}


To study the tradeoff between waits and aborts, we precisely control the workloads as two types: (1) we fix the first hotspot at the beginning of each transaction while moving the second around. In this case, the benefit \name gained over \ww is fixed and the chance of having cascading aborts increases as the distance between the two hotspots increases. We use the case to study how different magnitude of cascading aborts can affect the gains. (2) we fix the second hotspot at the end of the transactions and move the other around to study the case where the benefits and the chance of cascading aborts increase simultaneously for \name. 

\vspace{.05in}
\noindent\textbf{Fix One Hotspot at the Beginning}
\vspace{.05in}

\revised{In this experiment, we also show BAMBOO-base which does not have the second optimization of not retiring the last few operations. 
As shown in \figref{fig:hs2_fix0_th}, \name outperforms \ww for all distances. 
\figref{fig:hs2_fix0_runtime} shows how \name gains speedup by trading more aborts for less blocking. The improvements of \name can be up to 3$\times$. 
When the distance $x=0.75$ (i.e., there are 10 operations between the two hotspots), \name outperforms Wound-Wait by 37\%, although the abort rate of \name is 72\% higher. 
This result can be explained by the model, which indicates an improvement of at least 21.7\% in \name following $(A_{ww}-A_{bb})P_{\textit{conflict}} - B_{bb}*P_{\textit{cas\_abort}} = 15/16*1 - B_{bb}*0.72 \ge 0.217$.
The two versions of \name differ only when the second hotspot is at the end of the transaction ($x=1.0$). With the optimization, the last hotspot will not be retired which greatly reduces the bookkeeping overhead. \figref{fig:hs2_fix0_runtime} also illustrates that \name retains its saving from blocking while not suffering from aborts with the optimization. 
}


\begin{figure}[t]%
    \centering
    \begin{subfigure}[t]{0.03\columnwidth}
		\includegraphics[width=\linewidth]{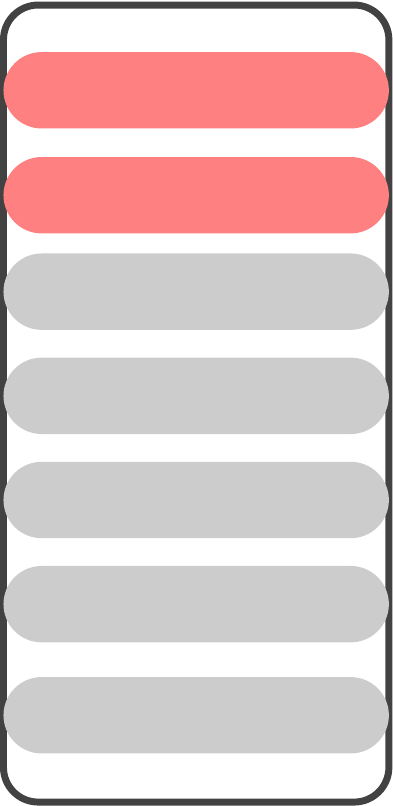}
    \end{subfigure}
    \begin{subfigure}[t]{0.45\columnwidth}
		\includegraphics[width=\linewidth]{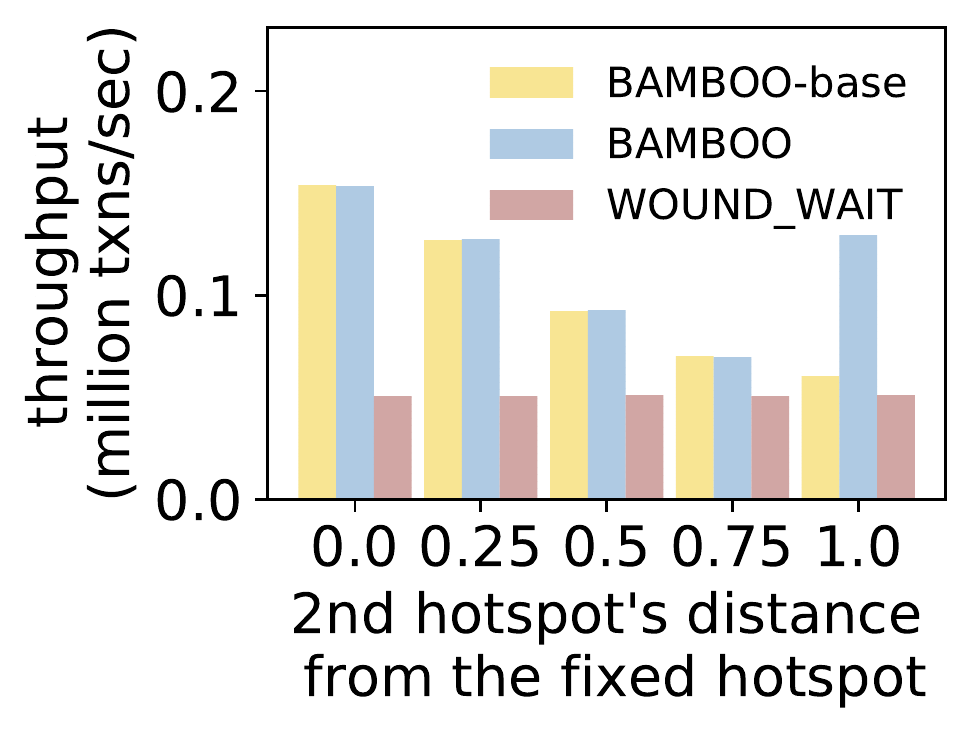}
		\caption{throughput}
		\label{fig:hs2_fix0_th} 
    \end{subfigure}
    \begin{subfigure}[t]{0.03\columnwidth}
		\includegraphics[width=\linewidth]{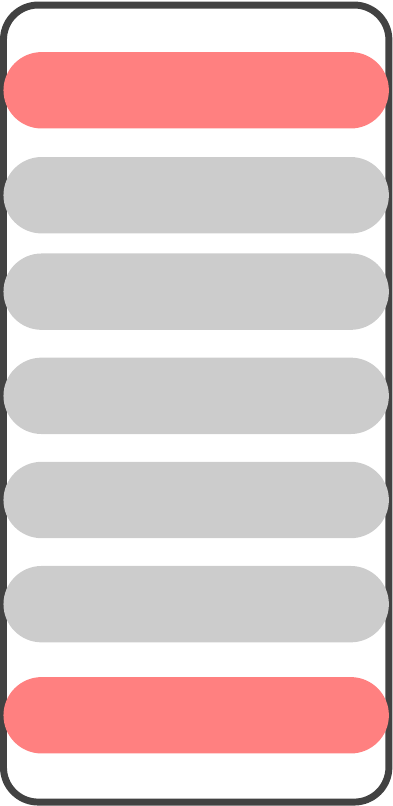}
    \end{subfigure}
    \begin{subfigure}[t]{0.45\columnwidth}
		\includegraphics[width=\linewidth]{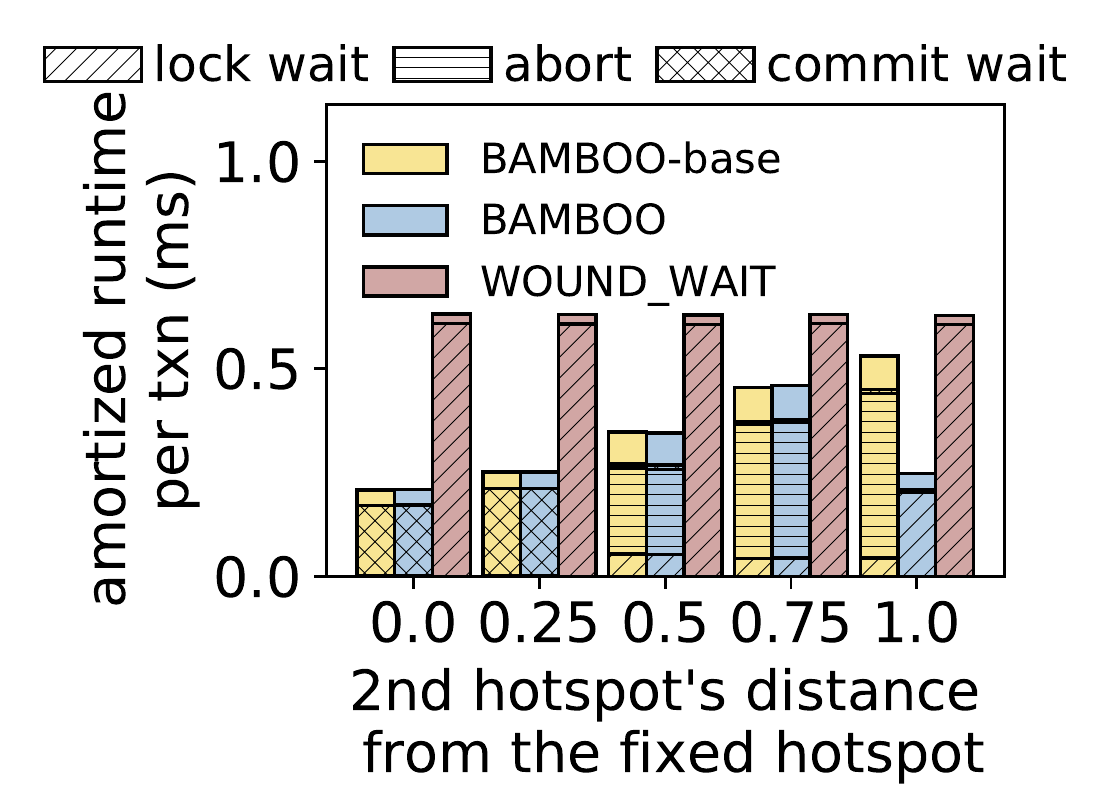} 
		\caption{runtime analysis}
    	\label{fig:hs2_fix0_runtime}
    \end{subfigure}
    \caption{\textbf{One hotspot at the beginning} for \name (left) and \ww (right), stored-procedure mode (32 threads)
    }
    \label{fig:hs2_fix0}%
\end{figure}

\vspace{.05in}
\noindent\textbf{Fix One Hotspot at the End}
\vspace{.05in}

We now fix the second hotspot at the end and change the position of the first hotspot. 
Compared to \figref{fig:hs2_fix1_th}, this workload here has less advantage for \name to begin with and yet introduces more cascading aborts as the benefit increases.
\revised{\figref{fig:hs2_fix1_runtime} shows that the time spent on aborts in \name never exceeds the time spent on waiting in \ww. However, \name without the second optimization (i.e., BAMBOO-base) may suffer from the overhead when it barely has benefits when $x = 0$, where the theoretical improvement is only 1/16.
We note that such cost is a function of the workload and underlying system. It may be significant with stored-procedure as shown here, which makes our optimization necessary in mitigating the problem. With other system setups such as the interactive mode we will show in the later experiments, the trade-off between such overhead and the gain from retiring can greatly change. 
}
%

 
\begin{figure}[t]%
    \centering
    \begin{subfigure}[t]{0.03\columnwidth}
		\includegraphics[width=\linewidth]{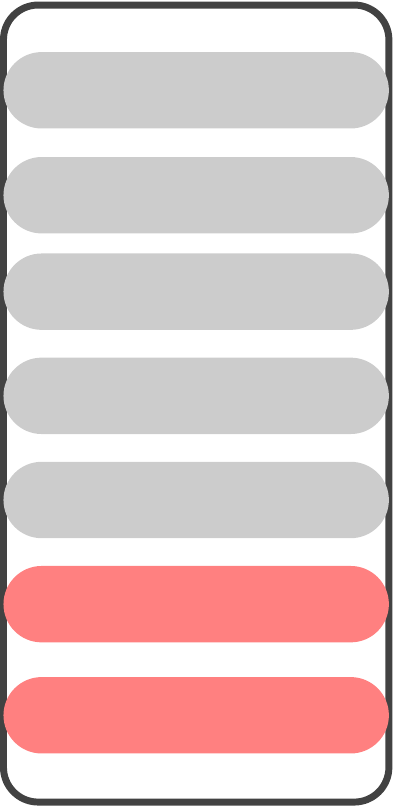}
    \end{subfigure}
    \begin{subfigure}[t]{0.45\columnwidth}
    	\center
		\includegraphics[width=\linewidth]{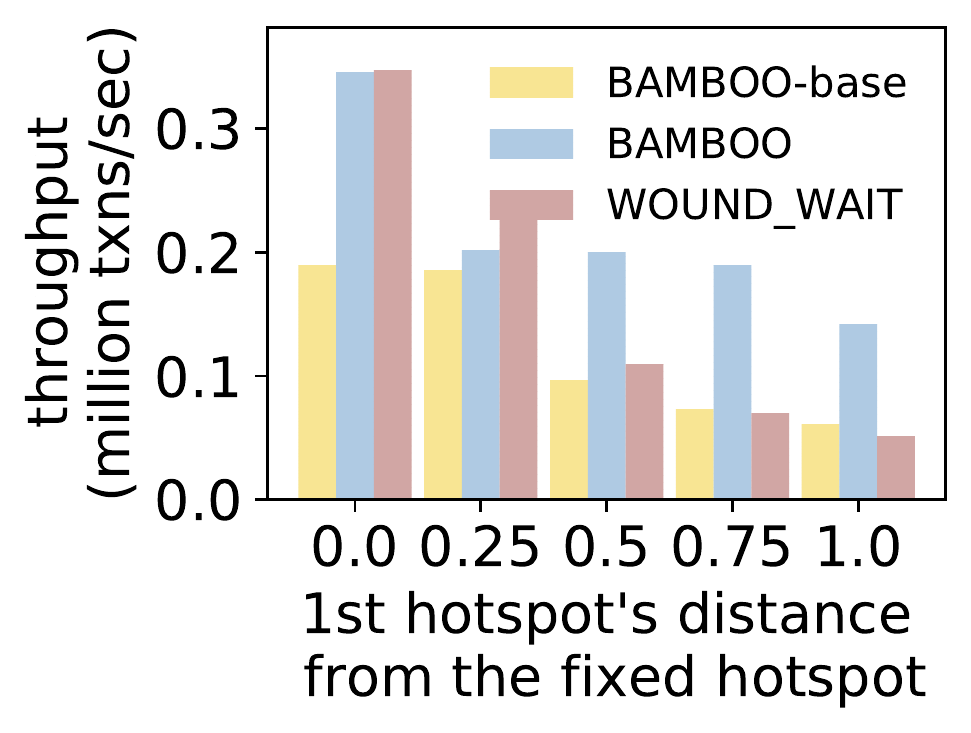}
		\caption{throughput}
		\label{fig:hs2_fix1_th} 
    \end{subfigure}
    \begin{subfigure}[t]{0.03\columnwidth}
		\includegraphics[width=\linewidth]{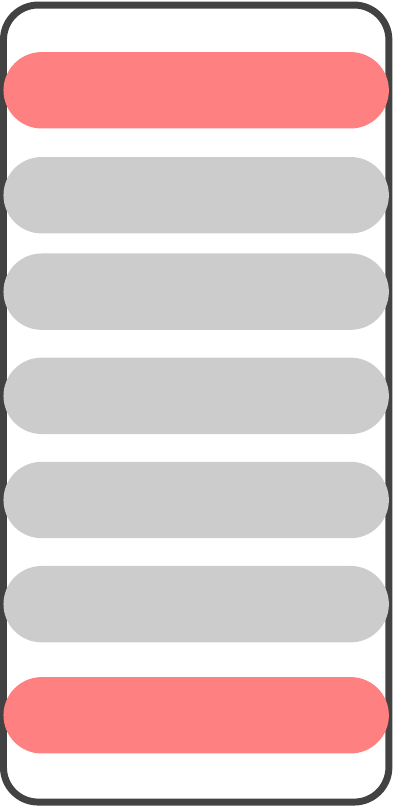}
    \end{subfigure}
    \begin{subfigure}[t]{0.45\columnwidth}
    	\center
		\includegraphics[width=\linewidth]{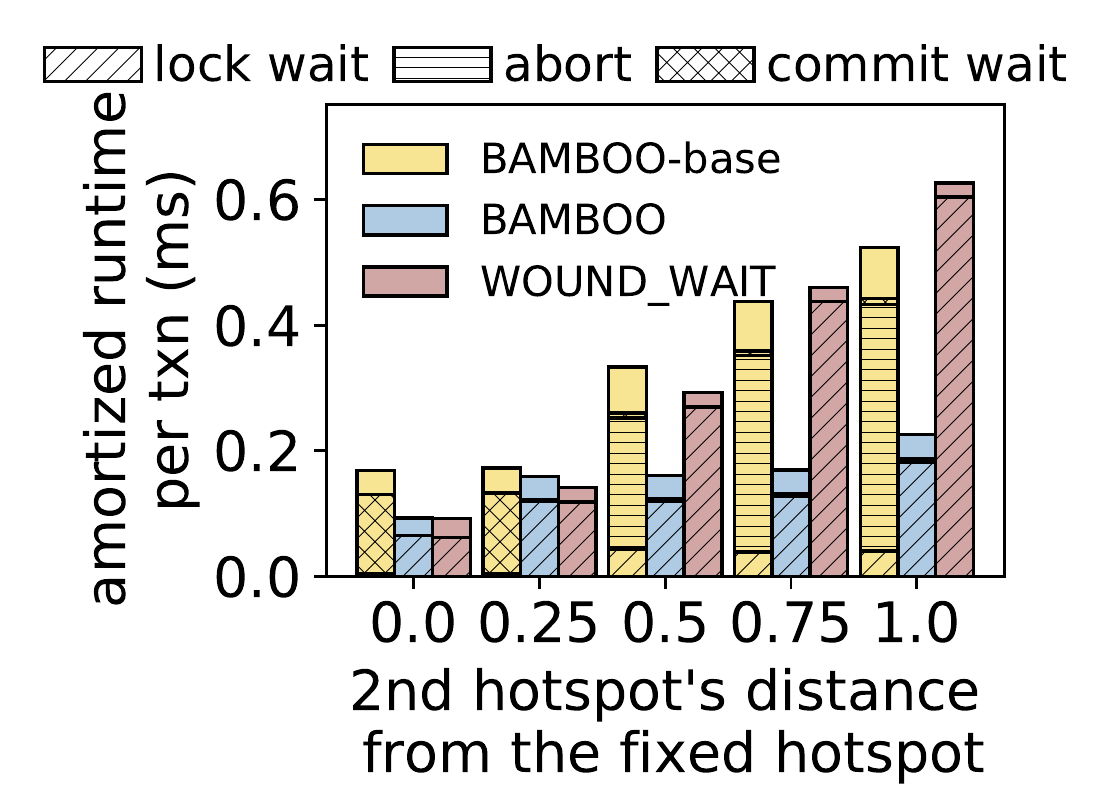} 
		\caption{runtime analysis}
    	\label{fig:hs2_fix1_runtime}
    \end{subfigure}
    \caption{\textbf{The second hotspot at the end} for \name (left) and \ww (right), stored-procedure mode (32 threads)}%
    \label{fig:hs2_fix1}%
\end{figure}

{\bf Summary:} The potential benefit of \name against \ww is a tradeoff between the benefit of reducing lock waiting time and the cost of cascading aborts and other overhead. Our measurements show that the benefits of reducing lock waiting time is usually greater than the cost of abort. However, \name can suffer from overhead with certain system setup when the benefit is minimal. In this case, our optimization of conditionally retiring some write operations should be applied for such cases. 

\vspace{-.15in}
\vspace{-.15in}
\subsection{Experiments on YCSB} \label{ssec:ycsb}
We now move to an even more complex workload --- YCSB with zipfian distribution. 
We will show how \name performs compared to other baselines as the number of threads, data accessing distribution, and read ratio vary. Note that in general \name only targets the high-contention setup in this workload as it is where hotspots (that most transactions would access) are present. 

The Yahoo! Cloud Serving Benchmark (YCSB)~\cite{cooper10} is a collection of workloads that are representative of large-scale services created by Internet-based companies. For all experiments in this section, we use a \revised{large scale database of more than 100 GB}, containing a single table with \revised{100 million} records. Each YCSB tuple has a single primary key column and then 10 additional columns each with 100 bytes of randomly generated string data. The DBMS creates a single hash index for the primary key. 

Each transaction in the YCSB workload by default accesses 16 records in the database. Each access can be either a read or an update. We control the overall read/write ratio of a transaction based on a specified $read\_ratio$.
We also control the workload contention level through adjusting $\theta$, a parameter controlling the Zipfian data distribution. For example, when $\theta = 0$, all tuples have equal chances to be accessed. When $\theta=0.6$ or $\theta=0.8$, a hotspot of 10\% of the tuples in the database are accessed by $\sim$40\% and $\sim$60\% of all transactions, respectively.

\begin{figure}[t]%
    \begin{subfigure}[t]{0.48\columnwidth}
    	\center
		\includegraphics[width=\linewidth]{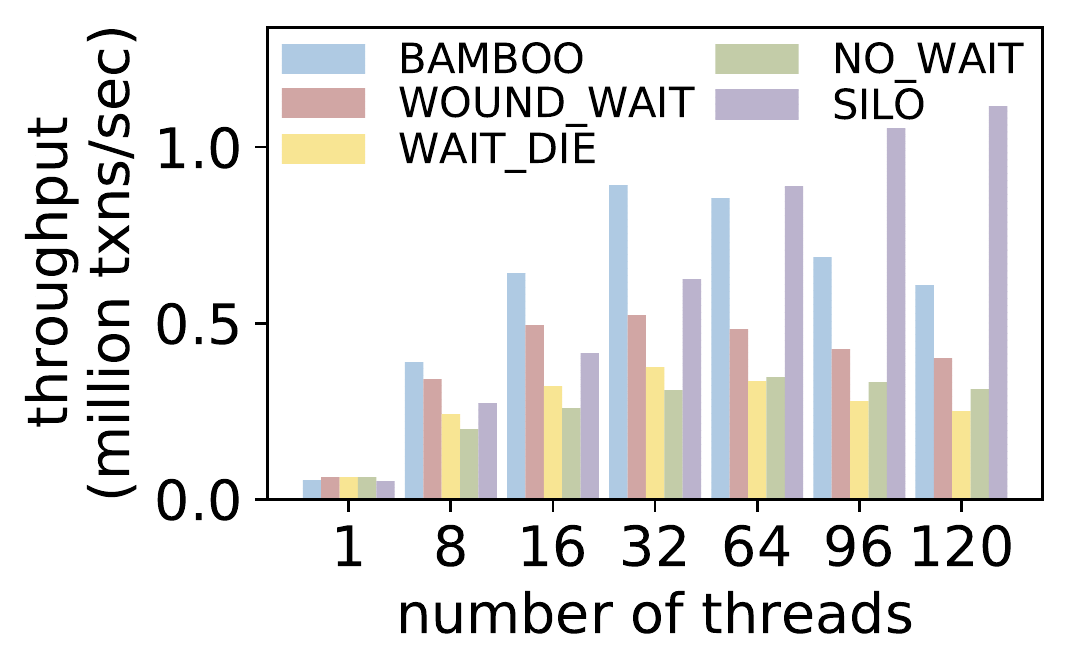}
		\caption{throughput}%
		\label{fig:ycsb_threads}%
    \end{subfigure}
    \begin{subfigure}[t]{0.48\columnwidth}%
    \center
    \includegraphics[width=\linewidth]{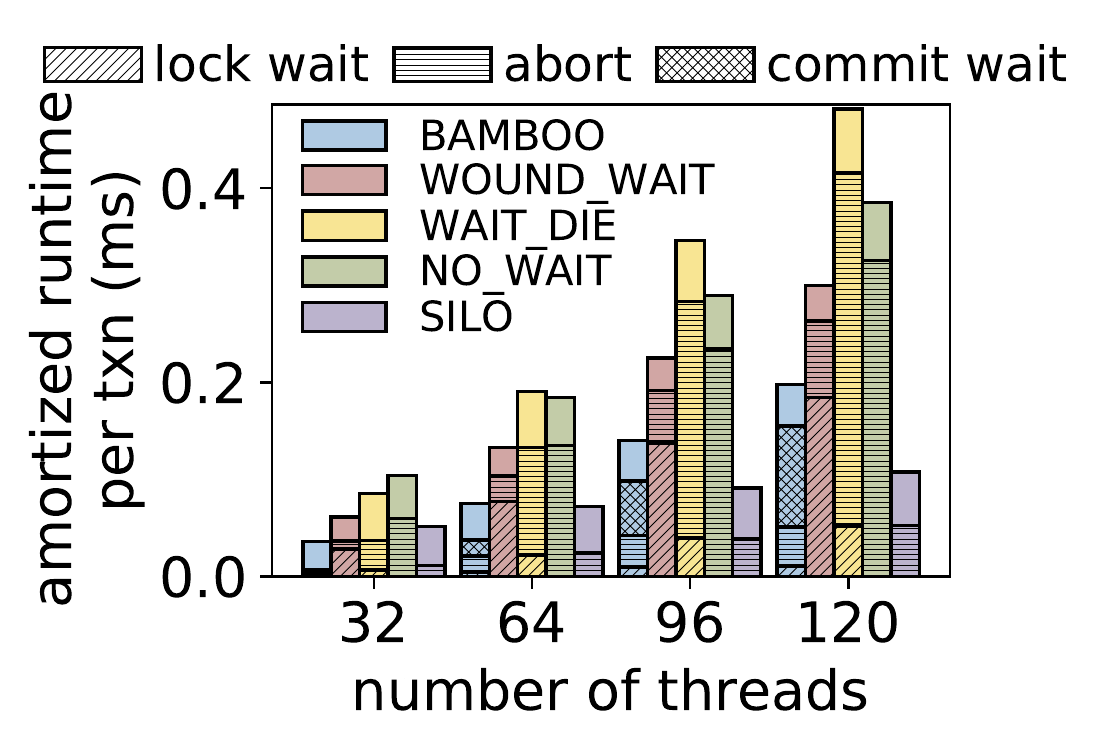} 
    \caption{runtime analysis}%
    \label{fig:ycsb_threads_runtime}%
	\end{subfigure}
	\vspace{-.1in}
	\caption{YCSB with varying thread count, stored-procedure mode ($\theta = 0.9$, $read\_ratio=0.5$)}
\end{figure}

\begin{figure}[t]%
	\centering
    \begin{subfigure}[t]{0.49\columnwidth}
    	\center
		\includegraphics[width=\linewidth]{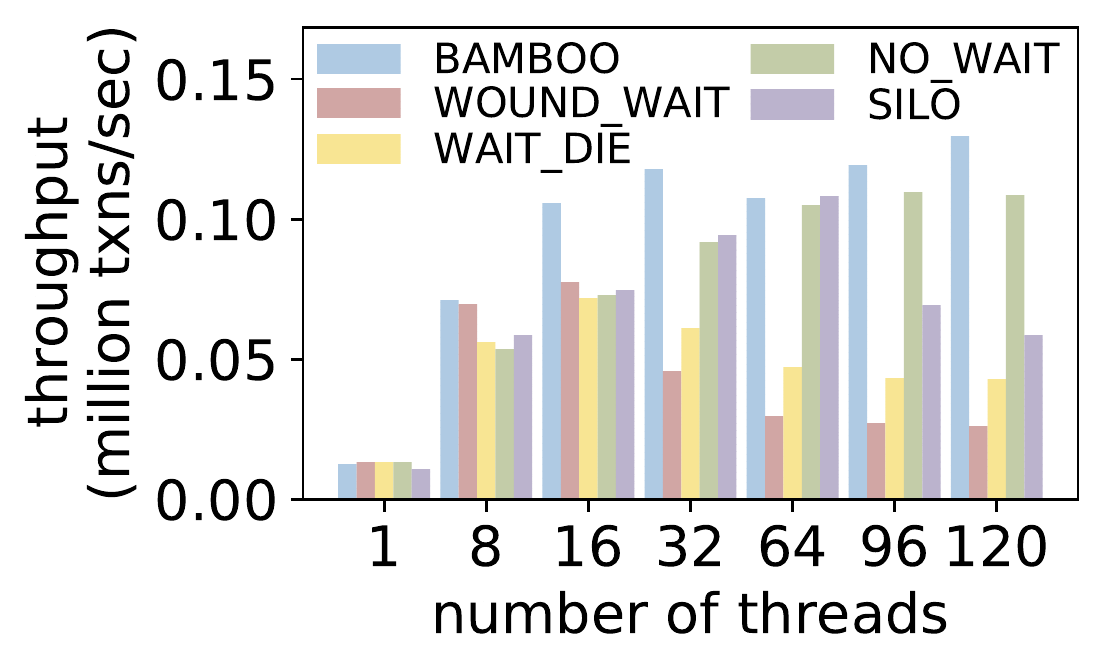}
		\caption{throughput}%
		\label{fig:ycsb_long_txn}%
    \end{subfigure}
    \begin{subfigure}[t]{0.49\columnwidth}%
    \center
    \includegraphics[width=\linewidth]{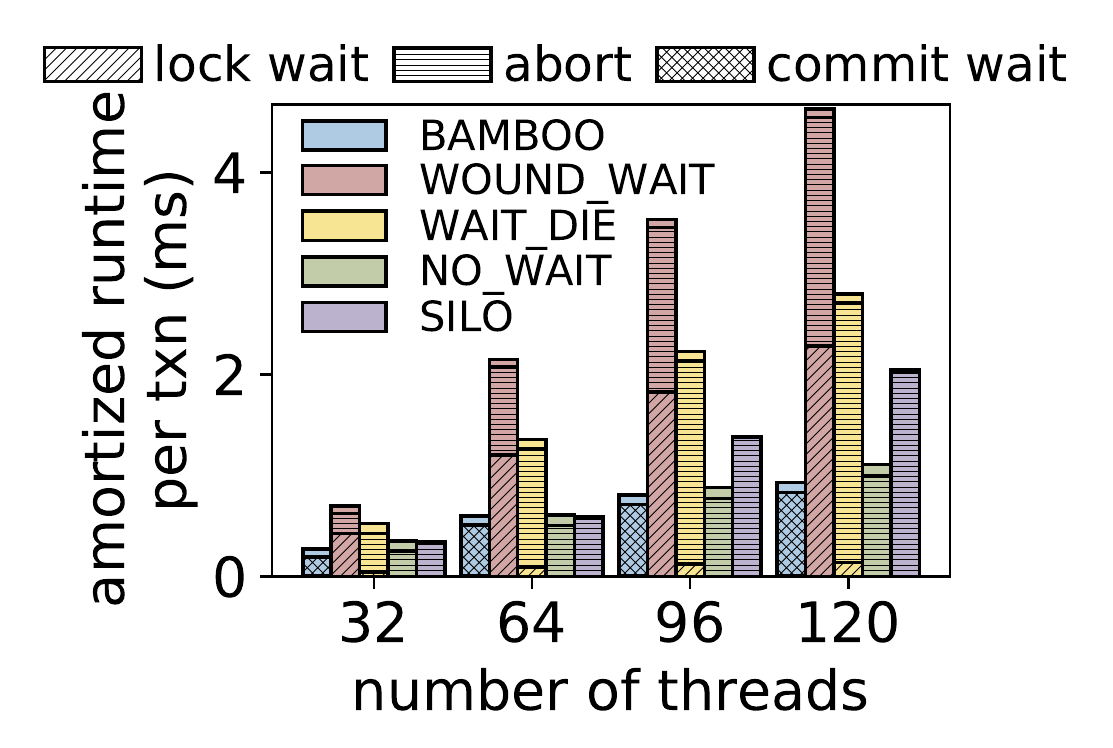} 
    \caption{runtime analysis}%
    \label{fig:ycsb_long_txn_runtime}%
	\end{subfigure}
	\vspace{-.1in}
	\caption{YCSB with 5\% long read-only transactions accessing 1000 tuples, stored-procedure mode ($\theta = 0.9$, $read\_ratio=0.5$) }
\end{figure}

{\bf Varying Number of Threads.} 
\revised{
\figref{fig:ycsb_threads} demonstrates that \name's improvement against \ww with different number of threads in highly contentious YCSB ($\theta=0.9$) configured in stored-procedure mode. \figref{fig:ycsb_threads_runtime} shows \name's benefits come from reducing waiting time without introducing many aborts. With 64 threads, \name achieves the maximum speedup against \ww, which is up to 1.77$\times$. 
All 2PL-based protocols show degradation after 32 threads and \name underperforms SILO when the thread count more than 96. This is mainly due to the intrinsic lock thrashing problem in 2PL~\cite{thomasian1993two}.} 

\revised{
{\bf Long Read-Only Transaction.} This experiment uses a workload with 5\% long read-only transactions accessing 1000 tuples and 95\% read-write transactions accessing 16 tuples. \figref{fig:ycsb_long_txn} shows that \name outperforms all other protocols most of the time. Compared with waiting-based protocols, \name benefits much from reducing waiting while rarely aborts. It shows an improvement of up to 5$\times$ against \ww. \name's optimization of no RAW conflict also contribute to the scenario as long read-only transactions will not block writes nor cause cascading aborts. SILO experience performance degradation in this case since long transactions may starve and aborts dominate the runtime as shown in \figref{fig:ycsb_long_txn_runtime}. \name also outperforms No-Wait as \name ensures the priorities of transactions and commit 20\% more long transactions than No-Wait when the thread count is 120. 
}

{\bf Varying Read Ratio.} We examine how varying $read\_ratio$ would influence the performance of different protocols in stored-procedure mode. 
\name shows improvements against all other protocols regardless of the read ratio. \revised{The percentage of improvement ranges from 27\% to 71\%}

{\bf Varying Data Accessing Distribution.} \figref{fig:ycsb_zipfian} shows how \name performs in both stored-procedure mode and interactive mode as $\theta$ of Zipfian distribution changes. 
As showed in \figref{fig:ycsb_zipfian_a}, \name outperforms all 2PL-based protocols under high contention (e.g. $\theta > 0.7$). \revised{Compared to \ww, \name provides up to 72\% improvements in throughput. For cases of lower contention, \name has $\sim$10\% degradation in throughput compared with \ww due to overhead. However, such degradation diminishes in the interactive mode where the expensive network communication dominates. In the interactive mode, \name is comparable with \ww with $\sim$8\% improvements when $\theta \leq 0.8$ and shows up to 2$\times$ speedup over \ww. 

Similarly, }SILO’s low-level performance optimizations (e.g., lock-free data structures) and the cache warming-up effect due to many aborts makes its performance significantly better than other 2PL-based protocols in stored-procedure mode. However, the performance advantage of Silo disappears in the interactive mode where aborts are significantly more expensive due to expensive gRPC calls.
\name outperforms all other protocols \revised{including SILO} in interactive mode where the cache warming-up effect has less impact. 

\begin{figure}[t]%
\vspace{.05in}
    \centering
    \begin{subfigure}[t]{0.49\linewidth}
    \center
    \includegraphics[width=0.97\linewidth]{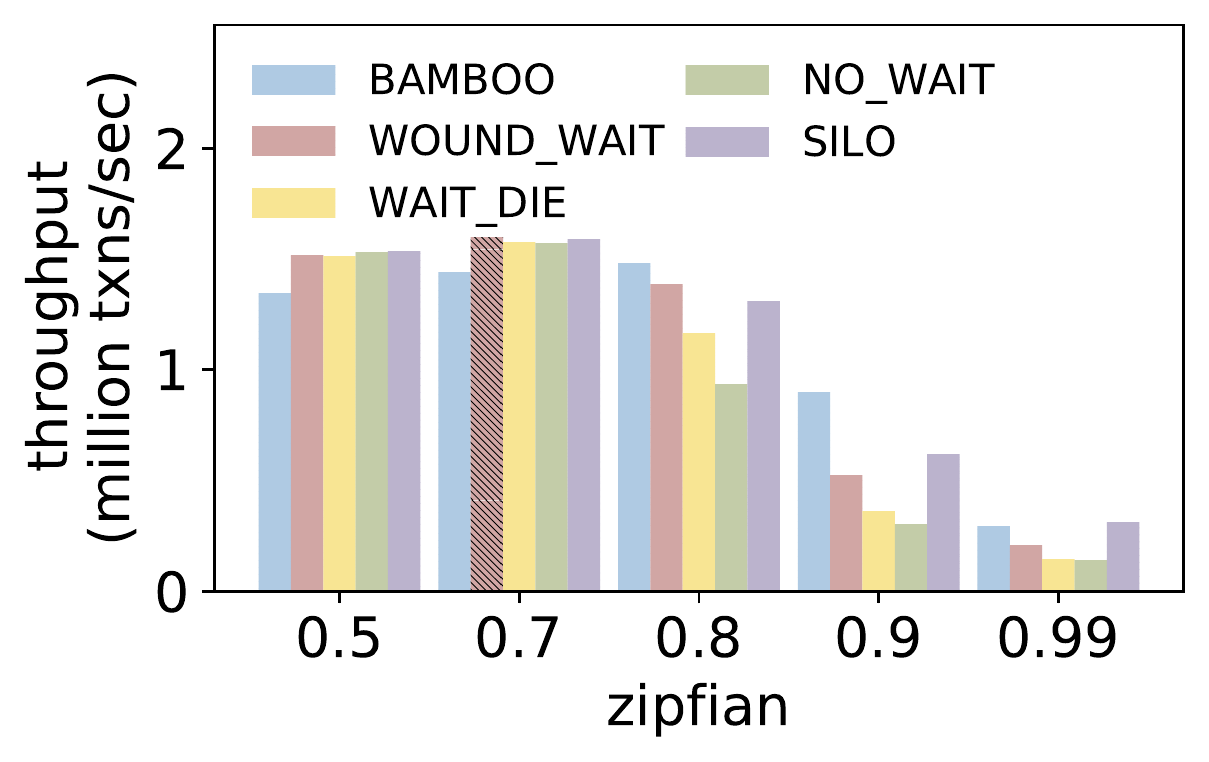}
    \caption{throughput}
    \label{fig:ycsb_zipfian_a}
    \end{subfigure}
    \begin{subfigure}[t]{0.49\linewidth}
    	\center
		\includegraphics[width=\linewidth]{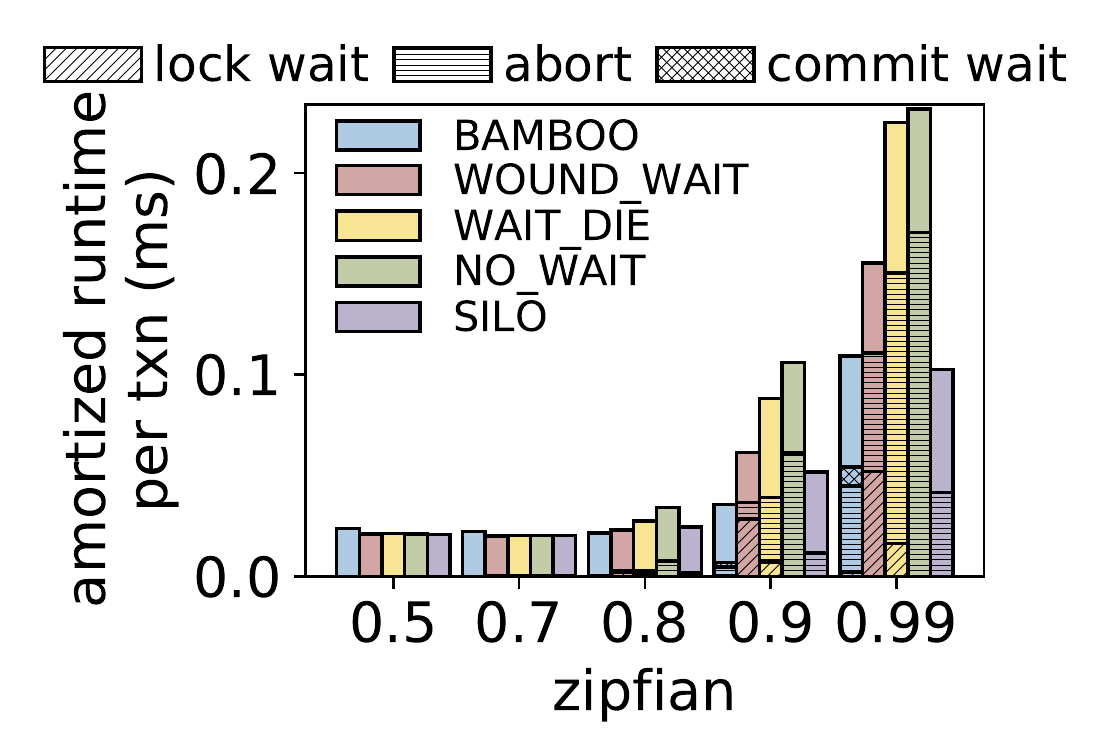} 
		\caption{runtime analysis}
		\label{fig:ycsb_zipfian_b}
    \end{subfigure}
    \vspace{-.1in}
    \caption{\textbf{YCSB with varying distribution} --- throughput vs. Zipfian skew level for YCSB, stored-procedure mode, ($read\_ratio = 0.5$, 16 threads)}
    \label{fig:ycsb_zipfian}
\end{figure}

\vspace{-.15in}

\subsection{Experiments on TPC-C Results}\label{ssec:tpcc}
Finally, we compare \name with other concurrency control schemes on the TPC-C benchmark~\cite{tpc-c}. We only ran experiments with 50\% new-order transactions and 50\% payment transactions. Note in the benchmark, 1\% of new order transaction are chosen at random to simulate user-initiated aborts. 

We first vary the number of threads. Figure~\ref{fig:tpcc_threads} presents the behavior of \name under high contention in stored-procedure mode. 
\name can obtain up to 2$\times$ improvements against \ww in stored-procedure mode. \revised{Similar to previous observations, SILO outperforms other 2PL protocols given its cache warming-up effect in stored-procedure mode.} In interactive mode, \name's performance scales up till 32 threads and achieves up to 4$\times$ and 14$\times$ improvements against \ww and SILO respectively.

\begin{figure}[t]%
    \centering
    \begin{subfigure}[t]{0.49\linewidth}
		\includegraphics[width=\linewidth]{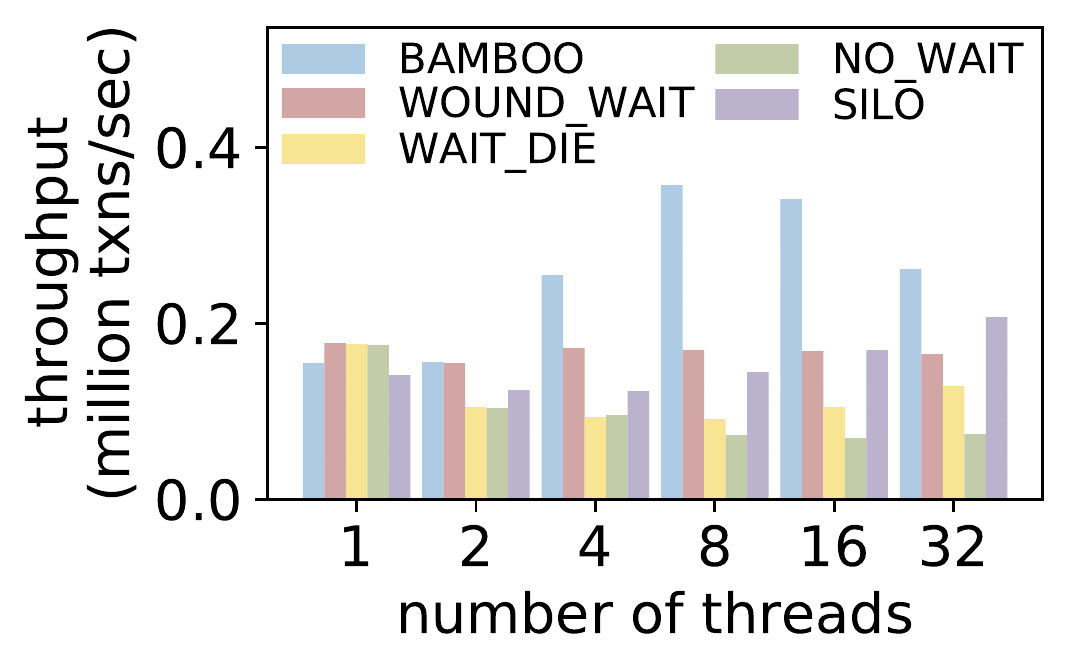} 
		\caption{stored-procedure mode}
    \end{subfigure}
    \begin{subfigure}[t]{0.45\linewidth}
		\includegraphics[width=\linewidth]{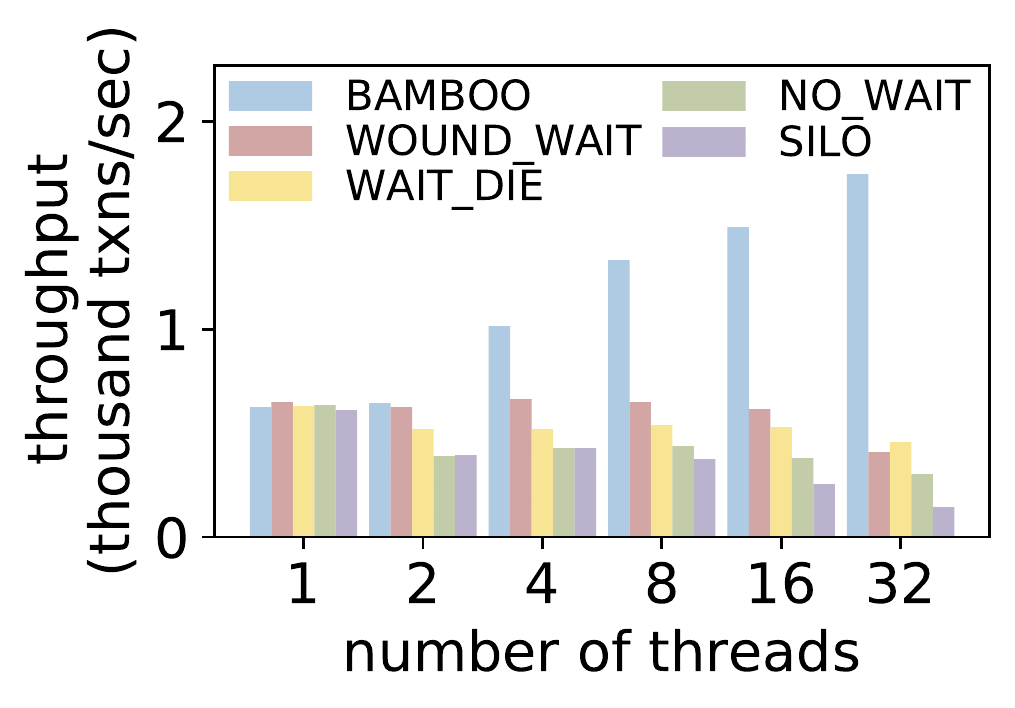} 
		\caption{interactive mode}
    \end{subfigure}
    \vspace{-.1in}
    \caption{vary \# of threads in TPC-C (1 warehouse)}%
    \label{fig:tpcc_threads}%
\end{figure}

\begin{figure}[t]%
    \centering
    \begin{subfigure}[t]{0.49\linewidth}
		\includegraphics[width=\linewidth]{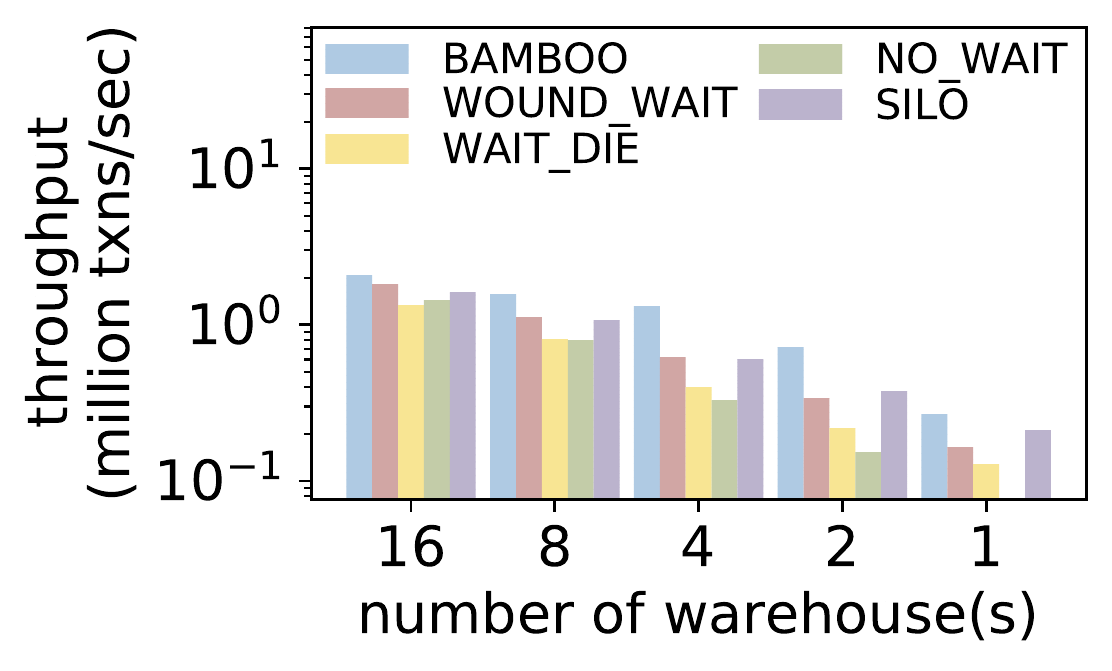} 
		\caption{stored-procedure mode}
    \end{subfigure}
    \begin{subfigure}[t]{0.45\linewidth}
		\includegraphics[width=\linewidth]{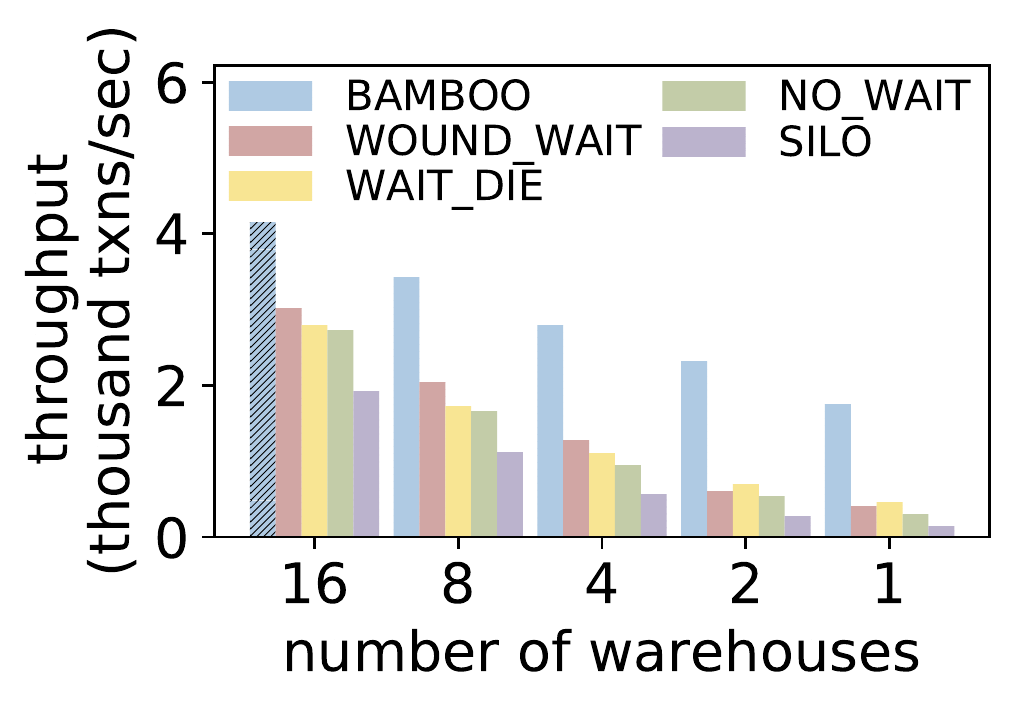} 
		\caption{interactive mode}
		\label{fig:tpcc_wh_b}
    \end{subfigure}
    \vspace{-.1in}
    \caption{vary \# of warehouses in TPC-C (32 threads)}%
    \label{fig:tpcc_wh}%
\end{figure}

Figure~\ref{fig:tpcc_wh} presents how \name performs with a different number of warehouses with 32 threads.
In stored-procedure mode, \name outperforms other 2PL-based protocols in \revised{high-contention} cases. The improvement depends on the number of warehouses. For example, when the number of warehouses is one (which is similar to the single hotspot cases), \name outperforms \ww by up to 2$\times$. When the workload is less contentious (e.g. with more warehouses), the difference between \name and other protocols is smaller. This is expected since \name targets at highly contentious cases. Figure~\ref{fig:tpcc_wh_b} shows that \name has more notable improvements of up to 4$\times$ over the best baseline when running in interactive mode.

\vspace{-.15in}

\subsection{Comparison with IC3}\label{sec:exp-ic3}

We choose to compare the performance of \name and IC3 in a separate section due the fact that IC3 requires the knowledge of the entire workload --- an assumption not made in the other protocols. 
We implemented IC3 with all its optimizations in DBx1000 and show the results with the best setting. Note that we omit the optimizations for commutative operations for a fair comparison to all algorithms. 

\figref{fig:tpcc_ic3_original} shows the comparison between \name and IC3 on the mix of payment and new-order transactions with a global warehouse table. As payment and new-order accesses different columns of warehouse table and district table (the most contentious tables in TPC-C), IC3's prior knowledge on all column accesses help it get rid of much contention. Given this, IC3 outperforms \name though it enforces some waiting when two transactions access the same column of different tuples.

However, for workloads where contentious transactions access the same columns of hotspot tuples, IC3 cannot gain such benefits and underperforms \name due to its column-level static analysis. To illustrate the effect, we simply modified new-order transactions to read one more column (\texttt{W\_YTD}) that will be updated by payment transactions. It is conceivable for a real-world transaction to read a hot field updated by other transactions. 
The result and runtime analysis are shown in \figref{fig:tpcc_ic3_modified} and \figref{fig:tpcc_ic3_modified_runtime}. While the performance of \name is barely affected, the performance of IC3 drops significantly. As expected, IC3 spends more time on waiting than \name due to its column-level static analysis. Note that, the increase in aborts is due to IC3's optimistic execution. The version without optimistic execution shows worse performance with more waiting. It would serialize the execution of potentially but not actually conflicting sub-transactions instead of just serializing the validation phases of these sub-transactions. Overall, \name has up \revised{1.5$\times$} improvement against IC3 on the slightly modified TPC-C workload that has ``true'' conflicts between payment and neworder transactions on the warehouse table. 

\begin{figure}[t]%
    \vspace{.1in}
    \begin{subfigure}[t]{0.49\linewidth}
		\includegraphics[width=\linewidth]{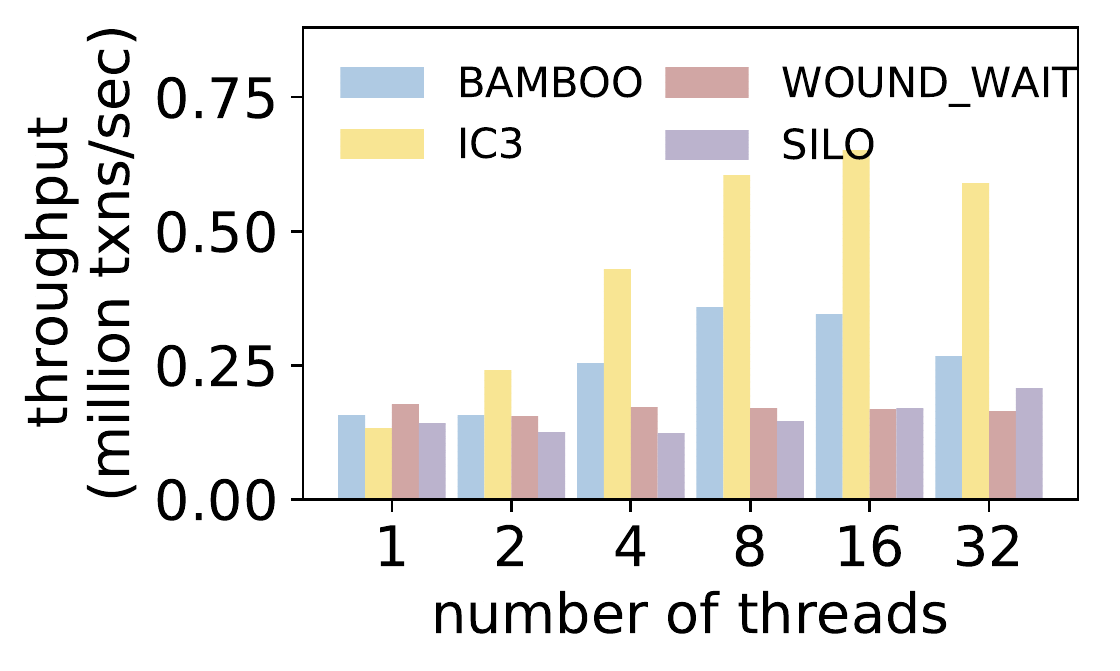} 
		\caption{with original new-order }
		\label{fig:tpcc_ic3_original}
    \end{subfigure}
    \begin{subfigure}[t]{0.45\linewidth}
		\includegraphics[width=\linewidth]{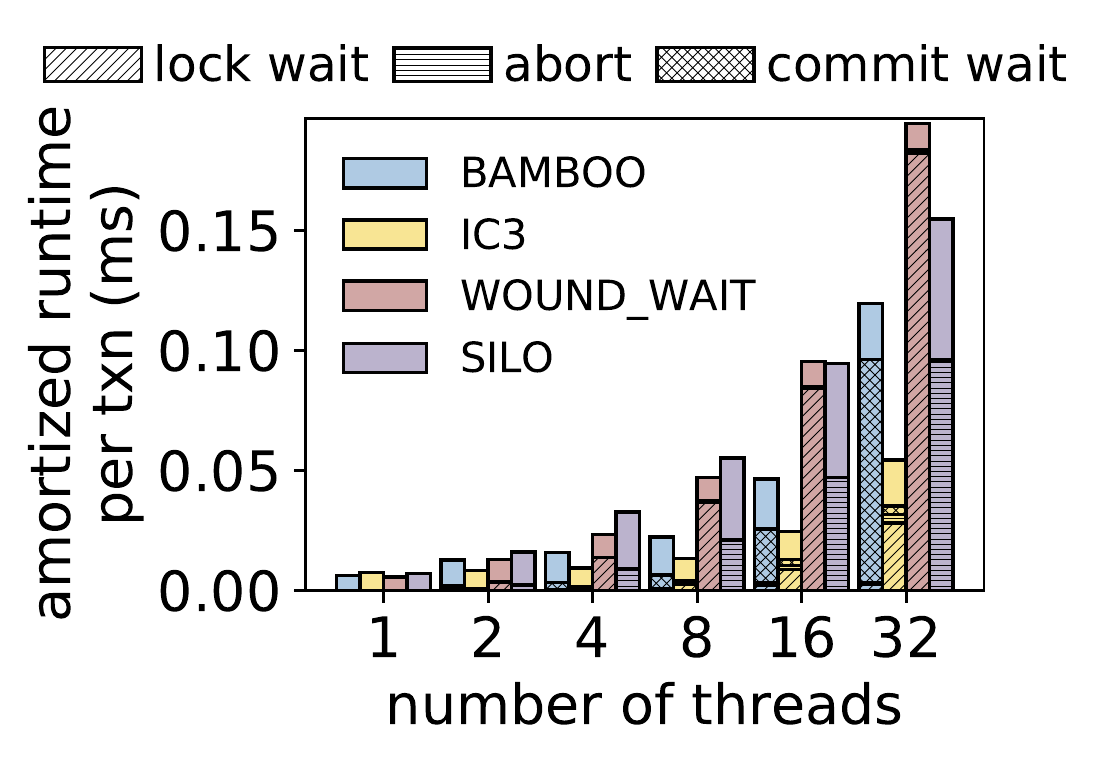} 
		\caption{runtime analysis}
    \end{subfigure}
    \vspace{-.1in}
    \begin{subfigure}[t]{0.49\linewidth}
		\includegraphics[width=\linewidth]{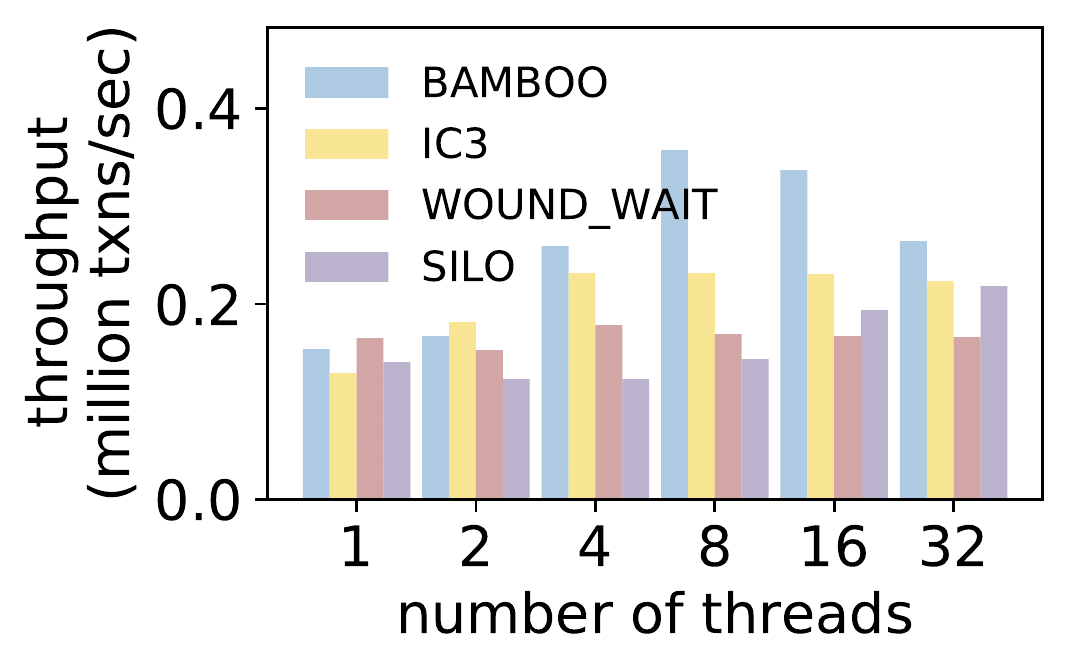} 
		\caption{with modified new-order}
		\label{fig:tpcc_ic3_modified}
    \end{subfigure}
    \begin{subfigure}[t]{0.45\linewidth}
		\includegraphics[width=\linewidth]{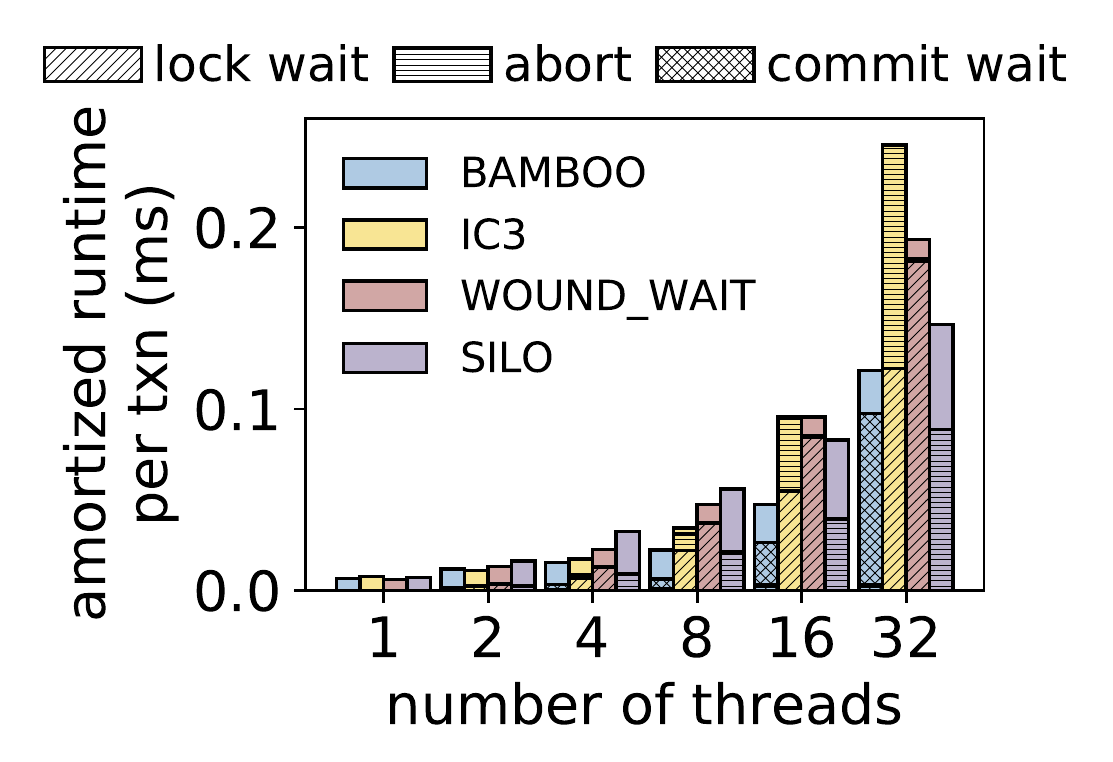} 
		\caption{runtime analysis}
		\label{fig:tpcc_ic3_modified_runtime}
    \end{subfigure}
    \caption{Bamboo vs. IC3 in TPC-C, stored-procedure mode (1 warehouse)}
    \label{fig:tpcc_ic3}%
\end{figure}

{\bf Summary.} In stored-procedure, \name is better than IC3 when there is a global hotspot that most transactions truly conflict on 
. The column-level static analysis in IC3 helps reduce contention when transactions are more likely to access different columns of the same tuple, however, it makes IC3 non-applicable for ad-hoc transactions. 
\vspace{-.15in}

\section{Related Work} \label{sec:related}
In this section, we discuss a few lines research related to \name. 
\subsection{Violating Two-Phase Locking}
Previous work has explored mechanisms that violate 2PL to improve transaction performance but targeted different aspects than \name. In distributed systems, Jones et al.~\cite{speculative-distributed} proposed a locking scheme that allows dependent transactions to execute when one transaction is waiting for its execution to finish in multiple partitions. The technique avoids making transactions wait for earlier transactions' distributed coordination. Gupta et al.~\cite{opt-distributed} proposed a distributed commit protocol where transactions are permitted to read the uncommitted writes of transactions in the prepare phase of two-phase commit, but not during transaction execution. 

Locking violation is also used to avoid holding locks while logging.
Early Lock Release (ELR)~\cite{kimura2012efficient, soisalon1995partial} is based the observation that a canonical transaction holds locks after the execution (i.e., has pre-committed) while waiting for the log to be flushed. ELR allows such pre-committed transactions to release their locks early so that other transactions can proceed and uses tags to enforce the commit order among dependent transactions. ELR has been applied in both research protocols~\cite{johnson2010aether} and commercial systems~\cite{orleans-txn}. 

Controlled lock violation (CLV)~\cite{graefe2013controlled} achieved the same goal as ELR. But instead of using tags to enforce the commit order, CLV extends data structures in the lock table to track and enforce the dependency order, therefore working for more general cases. The dependency tracking mechanism of \name is inspired by CLV. 
In contrast to ELR and CLV, \name explores violating 2PL to avoid holding locks on hotspots during the \textit{execution phase} before a transaction pre-commits to exploit more parallelism. 

\textit{Ordered shared locks}~\cite{agrawal1995ordered} explores violating 2PL in the execution phase but lacks specifications on key design components such as how to track dependency and avoiding deadlocks effectively. It has no qualitative or quantitative analysis on cascading aborts, a major concern of the approach, on modern systems. In contrast, this paper thoroughly analyzes the effect of cascading aborts and the inherent tradeoff between waits and aborts both qualitatively and quantitatively with proposed optimizations. We further discussed both key designs and new techniques in details such as safe retire with program analysis used in \name. In the end, we performed evaluations on modern systems comparing with state-of-the-art baselines.

\vspace{-.1in}
\subsection{Reading Uncommitted Data}

Previous work also proposed non-locking concurrency control protocols that can read uncommitted data. Faleiro et al.~\cite{faleiro2017high} proposed a protocol for deterministic databases that enables early write visibility, meaning that a transaction's writes are visible prior to the end of the execution. This protocol leverages the determinism where transaction execution is ordered prior to execution and the writes can be immediately visible to later transactions if the writing transaction will certainly commit. In contrast, \name does not 
rely on the assumption of determinism. 

Hekaton~\cite{diaconu13, larson11} proposed two protocols --- pessimistic version and optimistic version --- for main memory databases based on multiversioning. The pessimistic protocol allows for eager updates. Operations like appending updates to a read-locked data or reading the last-committed version of a write-locked data will not be blocked. If the owner of uncommitted dirty data is in \emph{preparing state}, dirty reads are allowed as well. However, as dirty data is not visible if its owner is in \emph{active state}, a write operation over write locked data by an active transaction will still be blocked. \name makes uncommitted data visible to reduce blocking time for transactions with write-after-write conflicts. Similarly to Hekaton, \name also tracks dependencies of transactions due to visible dirty data for serializable correctness, but in a different way. 

In addition to IC3, 
runtime pipelining~\cite{xie2015high} (RP) is another variant of transaction chopping. It is based on table-level static analysis combined with runtime enforcement. Specifically, they firstly derive a total ranking of all the read-write tables. Then it orders the sub-transactions based on the rank and enforces the execution to follow the order. However, sub-transactions still cannot be arbitrarily small to allow for more concurrency for two reasons. First, similar to IC3, accesses must be merged into one piece if they cause crossing of C-edges. 
Second, table-level analysis allows for less concurrency than column-level analysis.

Deferred runtime pipelining~\cite{mu2019deferred} (DRP) extends runtime pipelining to support both transactions where the access sets are known and unknown.
Similar to \name, DRP allows transactions to read \textit{tame} transactions' uncommitted data whenever the updates are done. However, DRP imposes stronger assumptions on the tame transactions that all the accesses must be known before execution to ensure serializability and being deadlock-free. DRP also introduces deferred execution to know when the updates are done and to reduce cascading aborts. However, the technique can be applied only when later operations do not depend on the previous ones. 


There are works redesigning the architecture to allow reading uncommitted data for hardware transaction. For example, Jeffrey et al. proposed Swarm~\cite{swarm1, swarm2} which divides a sequential program into many small ordered transactions and speculatively runs them in parallel. 
Unlike Swarm, \name is 
implemented in software and can be easily integrated into existing 2PL-based database systems.

\vspace{-.15in}
\subsection{Transaction Scheduling}

Previous work has investigated techniques to reschedule operations within a transaction for better performance. 

Quro~\cite{yan2016leveraging} changes the order of operations within a transaction to make hotspots appear as close to commit as possible to reduce the duration of locking period. However, Qura is subject to data dependency in the transaction thus often not able to flexibly move hotspots arbitrarily. \name does not require changing the order of transaction operators, but changes the concurrency control to handle hotspots, making it able to improve performance for a wider range of transactions than aforementioned work.

Ding et al.~\cite{ding2018improving} reorders the transactions within a batch to minimize inter-transaction conflicts and improve OCC for highly contentious cases. It models the problem of finding the best order while preserving the correctness as a feedback vertex set problem over directed graph. For each batch, they run a proposed greedy algorithm to approximate the solution of the NP-hard problem.
However, empirical evaluation shows the reordering process can take up to 17$\times$ of the transaction processing time. 
making the end-to-end performance lower than \name.

\section{Conclusion}
We proposed \name, a concurrency control protocol that extends traditional 2PL but allows the two-phase rule to be violated by retiring locks early. Through extensive analysis and performance evaluation, we demonstrated that \name can lead to significant performance improvement when the workload contain hotspots. Evaluation on TPC-C shows a performance advantage of up to 3$\times$. 

\bibliographystyle{abbrv}
\bibliography{bamboo,db}

\begin{thebibliography}{10}

\bibitem{dbx1000}
{DBx1000}.
\newblock \url{https://github.com/yxymit/DBx1000}.

\bibitem{code}
{DBx1000 with Bamboo Implemented}.
\newblock \url{https://github.com/ScarletGuo/DBx1000-Bamboo}.

\bibitem{agrawal1995ordered}
D.~Agrawal, A.~El~Abbadi, R.~Jeffers, and L.~Lin.
\newblock Ordered shared locks for real-time databases.
\newblock {\em The VLDB Journal}, 4(1):87--126, 1995.

\bibitem{bernstein1987concurrency}
P.~A. Bernstein.
\newblock {\em Concurrency control and recovery in database systems}, volume
  370.
\newblock Addison-wesley New York, 1987.

\bibitem{bernstein1981concurrency}
P.~A. Bernstein, P.~A. Bernstein, and N.~Goodman.
\newblock Concurrency control in distributed database systems.
\newblock {\em ACM Computing Surveys (CSUR)}, 13(2):185--221, 1981.

\bibitem{bernstein81}
P.~A. Bernstein and N.~Goodman.
\newblock {Concurrency Control in Distributed Database Systems}.
\newblock {\em CSUR}, pages 185--221, 1981.

\bibitem{bernstein2009principles}
P.~A. Bernstein and E.~Newcomer.
\newblock {\em Principles of transaction processing}.
\newblock Morgan Kaufmann, 2009.

\bibitem{bernstein79}
P.~A. Bernstein, D.~Shipman, and W.~Wong.
\newblock {Formal Aspects of Serializability in Database Concurrency Control}.
\newblock {\em TSE}, pages 203--216, 1979.

\bibitem{cooper10}
B.~F. Cooper, A.~Silberstein, E.~Tam, R.~Ramakrishnan, and R.~Sears.
\newblock {Benchmarking Cloud Serving Systems with YCSB}.
\newblock In {\em SoCC}, pages 143--154, 2010.

\bibitem{corbett12}
J.~C. Corbett and et~al.
\newblock {Spanner: Google's Globally-Distributed Database}.
\newblock In {\em OSDI}, pages 251--264, 2012.

\bibitem{diaconu13}
C.~Diaconu, C.~Freedman, E.~Ismert, P.-A. Larson, P.~Mittal, R.~Stonecipher,
  N.~Verma, and M.~Zwilling.
\newblock {Hekaton: {SQL} {S}erver's Memory-Optimized {OLTP} Engine}.
\newblock In {\em SIGMOD}, pages 1243--1254, 2013.

\bibitem{ding2018improving}
B.~Ding, L.~Kot, and J.~Gehrke.
\newblock Improving optimistic concurrency control through transaction batching
  and operation reordering.
\newblock {\em Proceedings of the VLDB Endowment}, 12(2):169--182, 2018.

\bibitem{cloudlab}
D.~Duplyakin, R.~Ricci, A.~Maricq, G.~Wong, J.~Duerig, E.~Eide, L.~Stoller,
  M.~Hibler, D.~Johnson, K.~Webb, A.~Akella, K.~Wang, G.~Ricart, L.~Landweber,
  C.~Elliott, M.~Zink, E.~Cecchet, S.~Kar, and P.~Mishra.
\newblock The design and operation of {CloudLab}.
\newblock In {\em Proceedings of the {USENIX} Annual Technical Conference
  (ATC)}, pages 1--14, July 2019.

\bibitem{orleans-txn}
T.~Eldeeb and P.~Bernstein.
\newblock {Transactions for Distributed Actors in the Cloud}.
\newblock Technical report, 2016.

\bibitem{eswaran76}
K.~P. Eswaran, J.~N. Gray, R.~A. Lorie, and I.~L. Traiger.
\newblock {The Notions of Consistency and Predicate Locks in a Database
  System}.
\newblock {\em CACM}, pages 624--633, 1976.

\bibitem{faleiro15}
J.~M. Faleiro and D.~J. Abadi.
\newblock {Rethinking Serializable Multiversion Concurrency Control}.
\newblock {\em {PVLDB}}, pages 1190--1201, 2015.

\bibitem{faleiro2017high}
J.~M. Faleiro, D.~J. Abadi, and J.~M. Hellerstein.
\newblock High performance transactions via early write visibility.
\newblock {\em Proceedings of the VLDB Endowment}, 10(5):613--624, 2017.

\bibitem{gawlick1985varieties}
D.~Gawlick and D.~Kinkade.
\newblock Varieties of concurrency control in ims/vs fast path.
\newblock {\em IEEE Database Eng. Bull.}, 8(2):3--10, 1985.

\bibitem{graefe2013controlled}
G.~Graefe, M.~Lillibridge, H.~Kuno, J.~Tucek, and A.~Veitch.
\newblock Controlled lock violation.
\newblock In {\em Proceedings of the 2013 ACM SIGMOD International Conference
  on Management of Data}, pages 85--96. ACM, 2013.

\bibitem{Gray1981ASM}
J.~Gray, P.~Homan, H.~Korth, and R.~Obermarck.
\newblock A straw man analysis of the probability of waiting and deadlock in a
  database system.
\newblock In {\em Berkeley Workshop}, 1981.

\bibitem{gray1992book}
J.~Gray and A.~Reuter.
\newblock {\em Transaction Processing: Concepts and Techniques}.
\newblock Morgan Kaufmann Publishers Inc., San Francisco, CA, USA, 1st edition,
  1992.

\bibitem{opt-distributed}
R.~Gupta, J.~Haritsa, and K.~Ramamritham.
\newblock Revisiting commit processing in distributed database systems.
\newblock In {\em SIGMOD}, page 486–497, 1997.

\bibitem{swarm1}
M.~C. Jeffrey, S.~Subramanian, C.~Yan, J.~Emer, and D.~Sanchez.
\newblock A scalable architecture for ordered parallelism.
\newblock In {\em MICRO}, page 228–241, 2015.

\bibitem{swarm2}
M.~C. {Jeffrey}, S.~{Subramanian}, C.~{Yan}, J.~{Emer}, and D.~{Sanchez}.
\newblock Unlocking ordered parallelism with the swarm architecture.
\newblock {\em IEEE Micro}, pages 105--117, 2016.

\bibitem{johnson2010aether}
R.~Johnson, I.~Pandis, R.~Stoica, M.~Athanassoulis, and A.~Ailamaki.
\newblock Aether: a scalable approach to logging.
\newblock {\em Proceedings of the VLDB Endowment}, 3(1-2):681--692, 2010.

\bibitem{speculative-distributed}
E.~P. Jones, D.~J. Abadi, and S.~Madden.
\newblock Low overhead concurrency control for partitioned main memory
  databases.
\newblock In {\em SIGMOD}, page 603–614, 2010.

\bibitem{kimura2012efficient}
H.~Kimura, G.~Graefe, and H.~A. Kuno.
\newblock Efficient locking techniques for databases on modern hardware.
\newblock In {\em ADMS@ VLDB}, pages 1--12, 2012.

\bibitem{kung1981optimistic}
H.-T. Kung and J.~T. Robinson.
\newblock On optimistic methods for concurrency control.
\newblock {\em ACM Transactions on Database Systems (TODS)}, 6(2):213--226,
  1981.

\bibitem{larson11}
P.-{\AA}. Larson, S.~Blanas, C.~Diaconu, C.~Freedman, J.~M. Patel, and
  M.~Zwilling.
\newblock {High-Performance Concurrency Control Mechanisms for Main-Memory
  Databases}.
\newblock {\em VLDB}, pages 298--309, 2011.

\bibitem{lomet2012multi}
D.~Lomet, A.~Fekete, R.~Wang, and P.~Ward.
\newblock {Multi-Version Concurrency via Timestamp Range Conflict Management}.
\newblock In {\em ICDE}, pages 714--725, 2012.

\bibitem{malkhi2013spanner}
D.~Malkhi and J.-P. Martin.
\newblock Spanner's concurrency control.
\newblock {\em ACM SIGACT News}, 44(3):73--77, 2013.

\bibitem{mohan1989aries}
C.~Mohan.
\newblock {\em {ARIES/KVL: A Key-Value Locking Method for Concurrency Control
  of Multiaction Transactions Operating on B-Tree Indexes.}}
\newblock VLDB, 1990.

\bibitem{mu2019deferred}
S.~Mu, S.~Angel, and D.~Shasha.
\newblock Deferred runtime pipelining for contentious multicore software
  transactions.
\newblock In {\em Proceedings of the Fourteenth EuroSys Conference 2019}, pages
  1--16, 2019.

\bibitem{dataflow-book}
F.~Nielson, H.~R. Nielson, and C.~Hankin.
\newblock {\em Principles of Program Analysis}.
\newblock Springer Publishing Company, Incorporated, 2010.

\bibitem{rosenkrantz1978system}
D.~J. Rosenkrantz, R.~E. Stearns, and P.~M. Lewis.
\newblock {System Level Concurrency Control for Distributed Database Systems}.
\newblock {\em ACM Transactions on Database Systems (TODS)}, 3(2):178--198,
  1978.

\bibitem{sadoghi2014}
M.~Sadoghi, M.~Canim, B.~Bhattacharjee, F.~Nagel, and K.~A. Ross.
\newblock Reducing database locking contention through multi-version
  concurrency.
\newblock {\em Proc. VLDB Endow.}, 7(13):1331–1342, Aug. 2014.

\bibitem{shasha1995transaction}
D.~Shasha, F.~Llirbat, E.~Simon, and P.~Valduriez.
\newblock {Transaction Chopping: Algorithms and Performance Studies}.
\newblock {\em ACM Transactions on Database Systems (TODS)}, 20(3):325--363,
  1995.

\bibitem{soisalon1995partial}
E.~Soisalon-Soininen and T.~Yl{\"o}nen.
\newblock Partial strictness in two-phase locking.
\newblock In {\em International Conference on Database Theory}, pages 139--147.
  Springer, 1995.

\bibitem{tang2018toward}
D.~Tang and A.~J. Elmore.
\newblock Toward coordination-free and reconfigurable mixed concurrency
  control.
\newblock In {\em 2018 $\{$USENIX$\}$ Annual Technical Conference
  ($\{$USENIX$\}$$\{$ATC$\}$ 18)}, pages 809--822, 2018.

\bibitem{tpc-c}
{The Transaction Processing Council}.
\newblock {TPC-C} {B}enchmark ({R}evision 5.9.0), June 2007.

\bibitem{thomasian1993two}
A.~Thomasian.
\newblock Two-phase locking performance and its thrashing behavior.
\newblock {\em ACM Transactions on Database Systems (TODS)}, 18(4):579--625,
  1993.

\bibitem{tu13}
S.~Tu, W.~Zheng, E.~Kohler, B.~Liskov, and S.~Madden.
\newblock {Speedy Transactions in Multicore In-Memory Databases}.
\newblock In {\em SOSP}, 2013.

\bibitem{wang2016mostly}
T.~Wang and H.~Kimura.
\newblock Mostly-optimistic concurrency control for highly contended dynamic
  workloads on a thousand cores.
\newblock {\em Proceedings of the VLDB Endowment}, 10(2):49--60, 2016.

\bibitem{wang2016scaling}
Z.~Wang, S.~Mu, Y.~Cui, H.~Yi, H.~Chen, and J.~Li.
\newblock Scaling multicore databases via constrained parallel execution.
\newblock In {\em Proceedings of the 2016 International Conference on
  Management of Data}, pages 1643--1658, 2016.

\bibitem{weikum2001transactional}
G.~Weikum and G.~Vossen.
\newblock {\em Transactional information systems: theory, algorithms, and the
  practice of concurrency control and recovery}.
\newblock Elsevier, 2001.

\bibitem{xie2015high}
C.~Xie, C.~Su, C.~Littley, L.~Alvisi, M.~Kapritsos, and Y.~Wang.
\newblock High-performance acid via modular concurrency control.
\newblock In {\em Proceedings of the 25th Symposium on Operating Systems
  Principles}, pages 279--294, 2015.

\bibitem{yan2016leveraging}
C.~Yan and A.~Cheung.
\newblock {Leveraging Lock Contention to Improve OLTP Application Performance}.
\newblock {\em Proceedings of the VLDB Endowment}, 9(5):444--455, 2016.

\bibitem{yu2014}
X.~Yu, G.~Bezerra, A.~Pavlo, S.~Devadas, and M.~Stonebraker.
\newblock {Staring into the Abyss: An Evaluation of Concurrency Control with
  One Thousand Cores}.
\newblock pages 209--220, 2014.

\bibitem{zhang2013transaction}
Y.~Zhang, R.~Power, S.~Zhou, Y.~Sovran, M.~K. Aguilera, and J.~Li.
\newblock Transaction chains: achieving serializability with low latency in
  geo-distributed storage systems.
\newblock In {\em SOSP}, pages 276--291, 2013.

\end{thebibliography}

\end{document}